\documentclass[lettersize,journal]{IEEEtran}
\usepackage{amsfonts}
\usepackage{graphicx}
\usepackage{amsmath, amssymb}
\usepackage{xcolor}  
\usepackage{bm}
\usepackage{algorithmic}
\usepackage{algorithm}
\usepackage{array}
\usepackage[caption=false,font=normalsize,labelfont=sf,textfont=sf]{subfig}
\usepackage{textcomp}
\usepackage{stfloats}
\usepackage{url}
\usepackage{verbatim}
\usepackage{graphicx}
\usepackage{cite}
\usepackage{amsthm}
\usepackage{array} 
\usepackage{arydshln}
\DeclareMathOperator*{\argmax}{argmax}
\DeclareMathOperator*{\argmin}{argmin}
\usepackage[caption=false]{subfig}

\newtheorem{remark}{\textbf{Remark}}

\newtheorem{proposition}{\textbf{Proposition}}

\begin{document}

\title{Pinching Antennas-Aided Integrated Sensing \\ and Multicast Communication Systems}

\author{Shan Shan, Chongjun Ouyang, Xiaohang Yang, Yong Li, Zhiqin Wang, and Yuanwei Liu
\thanks{Shan Shan and Yong Li are with the School of Information and Communication Engineering, Beijing University of Posts and Telecommunications, Beijing 100876, China (e-mail: \{shan.shan, liyong\}@bupt.edu.cn). 
Chongjun Ouyang is with the School of Electronic Engineering and Computer Science, Queen Mary University of London, London E1 4NS, U.K. (e-mail: c.ouyang@qmul.ac.uk). 
Xiaohang Yang and Zhiqin Wang are with China Academy of Information and Communications Technology, Beijing 100876, China (e-mail: \{yangxiaohang, wangzhiqin\}@caict.ac.cn).
Yuanwei Liu is with the Department of Electrical and Electronic Engineering, The University of Hong Kong, Hong Kong (e-mail: yuanwei@hku.hk).}
}
\IEEEaftertitletext{\vspace{-3em}}
\maketitle

\begin{abstract}
A pinching antennas (PAs)-aided integrated sensing and multicast communication framework is proposed. In this framework, the communication performance is measured by the multicast rate considering max-min fairness. Moreover, the sensing performance is quantified by the Bayesian Cram\'er-Rao bound (BCRB), where a Gauss-Hermite quadrature-based approach is proposed to compute the Bayesian Fisher information matrix. Based on these metrics, PA placement is optimized under three criteria: communications-centric (C-C), sensing-centric (S-C), and Pareto-optimal designs. These designs are investigated in two scenarios: the single-PA case and the multi-PA case. 1) For the single-PA case, a closed-form solution is derived for the location of the C-C transmit PA, while the S-C design yields optimal transmit and receive PA placements that are symmetric about the target location. Leveraging this geometric insight, the Pareto-optimal design is solved by enforcing this PA placement symmetry, thereby reducing the joint transmit and receive PA placement to the transmit PA optimization. 2) For the general multi-PA case, the PA placements constitute a highly non-convex optimization problem. To solve this, an element-wise alternating optimization-based method is proposed to sequentially optimize all PA placements for the S-C design, and is further incorporated into an augmented Lagrangian (AL) framework and a rate-profile formulation to solve the C-C and Pareto-optimal design problems, respectively. Numerical results show that: i) PASS substantially outperforms fixed-antenna baselines in both multicast rate and sensing accuracy; ii) the multicasting gain becomes more pronounced as the user density increases; and iii) the sensing accuracy improves with the number of deployed PAs.
\end{abstract}
\vspace{-10pt}
\vspace{-10pt}
\section{Introduction}
Integrated sensing and communications (ISAC) combines dual communication and sensing functionalities over shared wireless resources to improve spectrum utilization and system efficiency~\cite{9755276, 9737357}. 
Multiple-input multiple-output (MIMO) technology has been widely regarded as an efficient approach to ISAC, as it leverages spatial diversity and beamforming to simultaneously enhance spectral efficiency and sensing accuracy~\cite{10403776, 10038611}.
However, conventional MIMO architectures remain constrained by their geometrically fixed antenna configurations, which limit their adaptability to dynamic propagation environments.
To overcome this limitation, flexible-antenna systems have been proposed as an emerging antenna paradigm that enhance the spatial adaptability of wireless channels. 
Representative implementations include movable-antenna and fluid-antenna architectures~\cite{10839251,10707252}. Movable antennas change their physical positions to modify link geometry, while fluid antennas reshape their electromagnetic aperture through conductive-fluid redistribution. By locally adapting the propagation path, they can mitigate small-scale fading and improve both communication reliability and sensing accuracy in ISAC systems.
In addition, reconfigurable intelligent surfaces (RISs)~\cite{10319318} have also been introduced to mitigate propagation blockages. By leveraging controllable phase shifts of numerous passive reflecting elements, RIS technology establishes reliable virtual links between the base station (BS) and sensing targets.
However, the effectiveness of movable-antenna and fluid-antenna is generally constrained by limited movement ranges, which restricts their capability in addressing large-scale path-loss or line-of-sight (LoS) blockage.
Meanwhile, RIS encounters a severe double path-loss effect, especially at higher operating frequencies, which substantially degrades its reflection efficiency.

Recently, the \emph{Pinching-Antenna SyStem} (PASS) has been experimentally demonstrated as a practical realization of flexible-antennas that addresses the aforementioned limitations~\cite{fukuda2022pinching, 10945421}. PASS employs a dielectric waveguide as the transmission medium with low in-waveguide propagation loss, and its aperture length spans from a few meters to tens of meters. Along the waveguide, small dielectric elements, termed \emph{pinching antennas} (PAs), can be dynamically attached or detached, from which radio waves are transmitted or received.
A key advantage of PASS for ISAC lies in its scalable waveguide structure, which can be extended to be arbitrarily long. From a communication perspective, this establishes \emph{``near-wired''} links with strong LoS conditions to individual users. This characteristic effectively mitigates large scale path-loss and avoid LoS blockage~\cite{11036558}. Simultaneously, the long waveguide synthesizes a large effective aperture for sensing, which induces dominant near-field effects that facilitate precise polar domain localization~\cite{jiang2025nearfield}. 
\vspace{-10pt}
\subsection{Related Works}
The above advantages have motivated several early investigations into PASS-enabled communications and sensing. 
In particular, the authors in~\cite{11148125} provided an information-theoretic characterization of the achievable rate region for PASS-aided ISAC systems, which revealed a fundamental tradeoff between communication and sensing rates. Extending this analysis, the studies in~\cite{11122551, 11197530, mao2025multi} investigated the integration of PASS into ISAC systems, where the received signal-to-noise ratio (SNR) at the sensing targets was adopted as the performance metric. 
To obtain a more rigorous measure, subsequent studies adopted the Cram\'er-Rao bound (CRB), which provides the theoretical lower bound on the estimation error variance of any unbiased estimator~\cite{10217169, liu2022cramer}. 
Specifically, the CRB achieved by PASS is first derived in~\cite{ding2025pinchingisac} and then compared to that of conventional antennas, while~\cite{11111701} investigated its minimization via a particle swarm optimization (PSO)-based algorithm. 
In parallel, works such as \cite{bozanis2025cramer, li2025pinching, khalili2025pinching} analyzed round-trip sensing configurations using uniform linear arrays (ULAs) for echo reception. From an architectural perspective, the segmented waveguide system (SWAN) in~\cite{jiang2025segmented} enhances the degrees of freedom (DoF) available for optimization and characterizes the Pareto fronts for sensing and communication performance. 
Building on these theoretical foundations, recent works have extended the framework to diverse scenarios. To adapt to high-mobility environments, the work in~\cite{11288076} proposed a PASS-enabled ISAC framework for unmanned aerial vehicles (UAVs) utilizing deep reinforcement learning, while the authors in~\cite{11303890} addressed security concerns by jointly optimizing PA locations and artificial noise. Furthermore, to improve transmission efficiency, the work in~\cite{11314615} introduced an index modulation (IM)-based framework, which employs variational inference for joint user localization and symbol detection.
Collectively, these studies establish the theoretical foundation for PASS-aided ISAC systems and demonstrate the validity and effectiveness of PASS-enabled architectures.

\vspace{-10pt}
\subsection{Motivations and Contributions}
Despite the significant progress achieved by these works, two critical issues remain unresolved in PASS-aided ISAC research.
First, most existing sensing designs adopt the CRB as the metric under the idealized assumption that the parameters to be sensed are deterministic and known. In practice, however, the target parameters are typically unknown and random. Fortunately, their prior probability density functions (PDFs) can be obtained based on target properties and historical data. In such cases, the Bayesian CRB (BCRB) provides a more appropriate lower bound on the sensing mean-squared error (MSE) by explicitly exploiting the prior distribution information~\cite{kay1993estimation, 10584278, 10639496}. 
Second, regarding the communication functionality, existing PASS-enabled ISAC studies largely focus on single-user or unicast transmission. 
However, multicasting naturally aligns with the hardware characteristics of PASS. Since the waveguide is connected to a single RF chain, the available spatial DoF for mitigating inter-user interference are inherently limited, which makes this architecture more favorable for multicast transmission. By delivering a common message, multicasting avoids user-specific streams and thus further alleviates the DoF limitation of PASS-enabled multiuser transmission. Nevertheless, this important PASS-enabled multicast ISAC scenario remains largely unexplored.

To fill these research gaps and to obtain a deeper understanding of the communications-sensing tradeoff in PASS-enabled ISAC systems, we propose a PASS-aided integrated sensing and multicast communication framework. The major contributions of this paper are summarized as follows:

\begin{itemize}
    \item We propose a PASS-aided ISAC framework where the PA placement is optimized to perform simultaneous information multicasting and target sensing. Within this framework, the multicast communication performance is quantified by the max-min fairness (MMF) rate, while the BCRB is adopted to characterize the relationship between the PA placement and sensing accuracy. Since the BCRB evaluation involves intractable integrals, the Gauss-Hermite quadrature (GHQ) rule is employed to facilitate the computation of the Fisher information matrix (FIM). Based on these metrics, the PA placement optimization is then investigated under three design criteria, namely: 1) communications-centric (C-C) design, which maximizes the multicast rate, 2) sensing-centric (S-C) design, which minimizes the BCRB, and 3) Pareto-optimal design, which characterizes the communications-sensing tradeoff.
    \item We start from a single-PA setup to reveal key design insights and to characterize the Pareto boundary. i) For the C-C design, we derive a closed-form solution in which the optimal transmit PA lies in a finite candidate set. ii) For the S-C design, we derive that under one-dimensional target location uncertainty, the optimal transmit and receive PA placements admit a symmetric layout. iii) For the Pareto-optimal design, we adopt a rate-profile-based formulation to characterize the tradeoff boundary. By exploiting the fact that the multicast rate is independent of the receive PA, the receive PA placement follows a closed-form sensing-driven symmetry rule, and the remaining transmit PA placement reduces to an efficient univariate optimization that can be solved via a simple one-dimensional search.
    \item 
    We further extend our study to the multi-PA setting, where a more rigorous design is considered by explicitly accounting for performance constraints beyond the objective. i) For the S-C design, we propose an element-wise alternating optimization (AO) method to sequentially optimize all PA placements so as to minimize the BCRB under a minimum multicast rate constraint. ii) For the C-C design, we develop a penalty-based adaptation to maximize the multicast rate subject to a maximum BCRB constraint. iii) For the Pareto-optimal design, we propose a rate-profile-based scalarization approach to characterize the tradeoff boundary and solve the resulting problem using the sequential element-wise AO algorithm.
    \item Finally, we provide numerical results to validate the effectiveness of the proposed algorithms. Our findings reveal that: i) PASS achieves substantial performance gains over conventional fixed-location antenna systems in both multicast communication and sensing performance. ii) For multicasting, the performance advantage of PASS becomes more pronounced in scenarios with high user density. iii) For sensing, the deployment of additional PAs yields increasingly precise target localization.
\end{itemize}

The reminder of this paper is organized as follows. Section~II presents the PASS-ISAC system model and performance metrics. Section~III and~IV optimize the PA placement in a simplified single-PA scenario and a more general multi-PA scenario, respectively. Section~V provides numerical results. Finally, Section~VI concludes the paper.

\subsubsection*{Notations}
Scalars, vectors, and matrices are represented by regular, bold lowercase, and bold uppercase letters, respectively. The sets of complex and real numbers are denoted by $\mathbb{C}$ and $\mathbb{R}$. The operators $(\cdot)^{-1}$, $(\cdot)^{\ast}$, $(\cdot)^{\mathsf T}$, $(\cdot)^{\mathsf H}$, and ${\mathsf {tr}}(\cdot)$ correspond to the inverse, conjugate, transpose, conjugate transpose, and trace, respectively. For a vector ${\bf x}$, $[{\bf x}]_i$ represents its $i$th element, while ${\mathsf{Diag}}({\bf x})$ constructs a diagonal matrix using the elements of ${\bf x}$. The Kronecker delta is denoted by $\delta(x,y)$, which equals $1$ if $x = y$ and $0$ otherwise. The operator $[x]^+ \triangleq \max\{x, 0\}$ represents the non-negative projection of $x$. ${\mathcal C}{\mathcal N}(a, b^2)$ signifies a circularly symmetric complex Gaussian distribution with mean $a$ and variance $b^2$. The statistical expectation is given by ${\mathbb{E}}\{\cdot\}$. Furthermore, $|\cdot|$ and $\|\cdot\|$ denote the absolute value and Euclidean norm, respectively. Finally, $\Re \{\cdot\}$ extracts the real part of a complex number, and ${\mathcal O}(\cdot)$ denotes the big-O notation.

 \begin{figure}[!t]
\centering
\includegraphics[height=0.23\textwidth]{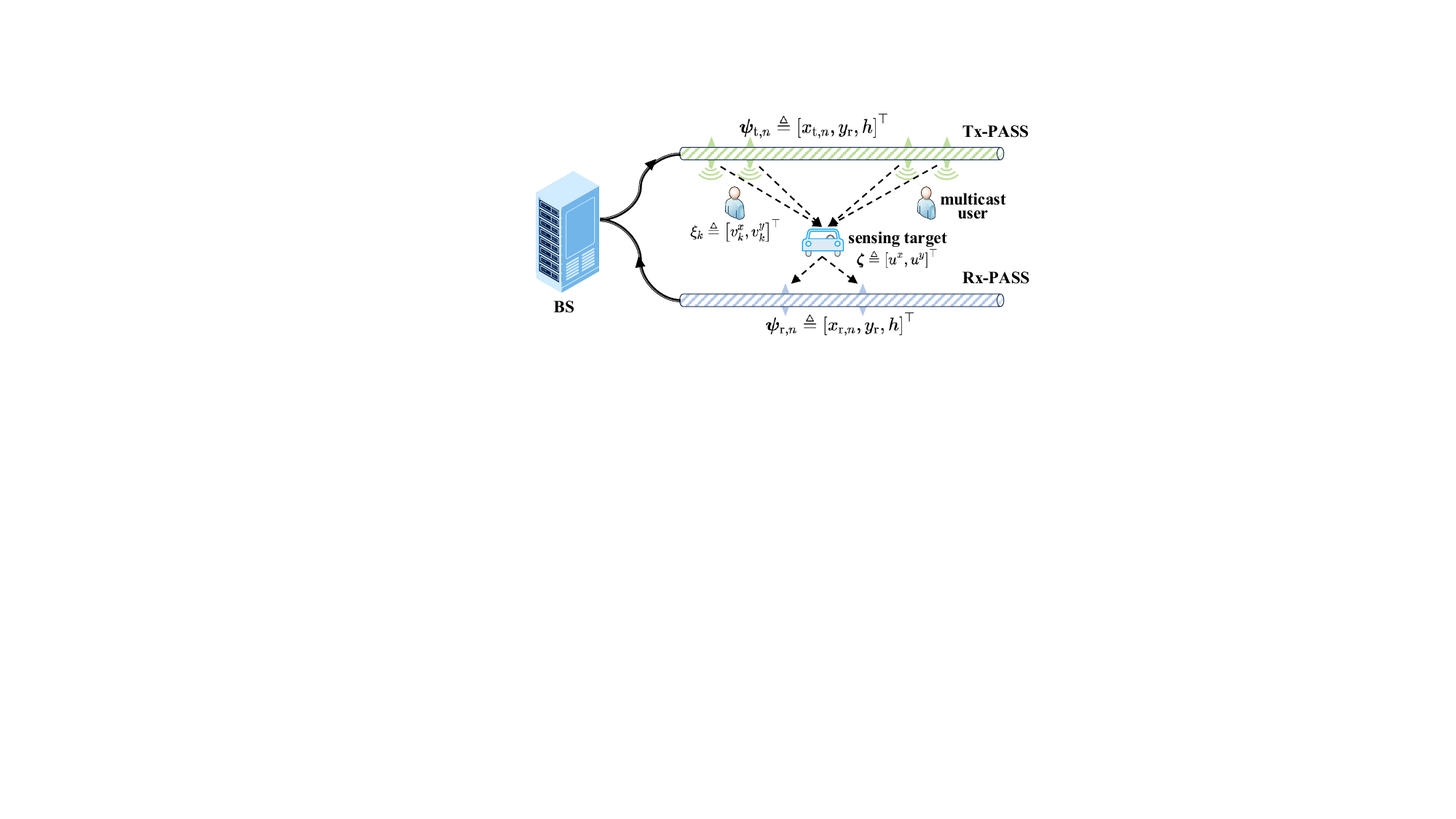}
\caption{Illustration of PA-aided integrated sensing and multicast transmission.}
\label{Fig_1}
\end{figure}
\vspace{-10pt}
\section{System Model}
We consider a PASS-aided ISAC system that simultaneously performs information multicasting and sensing. The dual-functional BS is equipped with a transmit waveguide (\emph{Tx-PASS}) employing $N_{\rm t}$ Tx-PAs and a receive waveguide (\emph{Rx-PASS}) employing $N_{\rm r}$ Rx-PAs. Both waveguides and their associated PAs are deployed at a fixed height $h$ above the ground, within a rectangular coverage area of size $D_{\rm x} \times D_{\rm y}$. For practical deployment, the Tx-PASS and Rx-PASS are aligned parallel to the $x$-axis. The feed points of the Tx-PASS and Rx-PASS are located at ${\bm\psi}_{{\rm t},0} \triangleq \left[0, y_{\rm t}, h\right]^{\mathsf T}$ and ${\bm\psi}_{{\rm r},0} \triangleq \left[0, y_{\rm r}, h\right]^{\mathsf T}$, respectively, where $y_{\rm t}$ and $y_{\rm r}$ denote their positions along the $y$-axis. Define ${\mathcal N}_{\rm t} \triangleq \{1, \dots, N_{\rm t}\}$ and ${\mathcal N}_{\rm r} \triangleq \{1, \dots, N_{\rm r}\}$ as the index sets of Tx-PAs and Rx-PAs, respectively. The location of the $n$th Tx-PAs and Rx-PAs are denoted by ${\bm\psi}_{{\rm t},n} \triangleq \left[x_{{\rm t},n}, y_{\rm t}, h\right]^{\mathsf T}, \forall n\in{\mathcal N}_{\rm t}$ and ${\bm\psi}_{{\rm r},n} \triangleq \left[x_{{\rm r},n}, y_{\rm r}, h\right]^{\mathsf T},  \forall n\in{\mathcal N}_{\rm r}$, respectively.
Let ${\bf x}_{\rm t} \triangleq [x_{{\rm t},1}, x_{{\rm t},2}, \dots, x_{{\rm t},N_{\rm t}}]^{\mathsf T}$ collect the $x$-coordinate of all Tx-PAs. These coordinates satisfy the spatial constraint $0 \le x_{{\rm t},1} < \dots < x_{{\rm t},N_{\rm t}} \le D_{\rm x}$ and a minimum inter-PA spacing constraint $|x_{{\rm t},n} - x_{{\rm t},n-1}| \ge \Delta_{\min} = \frac{\lambda}{2}, \ \forall n \in {\mathcal N}_{\rm t},~n \ge 2$, where $\lambda$ is the free-space wavelength~\cite{10981775}. The $x$-coordinates of the Rx-PAs follow the same constraints. 

The considered ISAC system simultaneously serves $K$ single-antenna user and senses a single target, as illustrated in Fig.~\ref{Fig_1}. Let $\mathcal{K}\triangleq\{1,2,\ldots,K\}$ denote the index sets of the users.  The position of the $k$th user is denoted by $\boldsymbol\xi_k\triangleq[v_k^x,v_k^y, 0]^{\mathsf T}$, whereas the target is assumed to be located at $\boldsymbol\zeta\triangleq[u^x,u^y, 0]^{\mathsf T}$.
The BS simultaneously transmits common information signals to all users and utilizes the reflected echo signals for sensing. Moreover, we focus on the fundamental sensing process and neglect the potential interference of downlink communication signals on the received echoes, which assumes that such coupling effects can be effectively mitigated through standard self-interference cancellation and synchronization techniques~\cite{9737357}.
Let ${s}\in\mathbb{C}$ denote the ISAC transmit signal emitted by the BS. The BS is subject to a maximum transmit power budget of $P_{\rm t}$, where $P_{\rm t}\neq 1$. To achieve the multicast capacity limit under this power constraint, the transmitted signal is ${s} = \sqrt{P_{\rm t}}s_{\rm o}$, where $s_o\sim{\mathcal{CN}}(0,1)$ denotes the normalized information symbol.
\subsection{Multicast Communication Performance Metric}
We first consider the multicast communication channel. Let ${\bf h}({\boldsymbol{\xi}}_{k}, {\bf x}_{{\rm t}})\in\mathbb{C}^{N_{\rm t}\times1}$ denote the channel vector from the BS to user $k\in\mathcal{K}$. The signal received by the $k$th user is
\begin{align}
{y}_{k} = {\bf h}^{\mathsf T}({{\boldsymbol{\xi}}_k}, {\bf x}_{{\rm t}}){\boldsymbol{\phi}}_{\rm t}({\bf x}_{\rm t}){s} + {z}_{k}, \quad \forall k\in {\mathcal K}, \label{eq:Received com signal}
\end{align}
where ${z}_{k}\sim\mathcal{CN}(0,\sigma_{k}^2)$ represents the additive noise at the $k$th user with noise variable $\sigma_k^{2}$.
The $n$th element of ${\bf h}({{\boldsymbol{\xi}}_k},{\bf x}_{{\rm t}})$ represents the free-space channel coefficient between the $n$th Tx-PA and the $k$th user, which can be expressed as follows:
\begin{align}\label{eq:element_h}
\left[{\bf h}({{\boldsymbol{\xi}}_k}, {\bf x}_{{\rm t}})\right]_n\triangleq\frac{\sqrt{\eta}{\rm{e}}^{{\rm{-j}}k_0\lVert{{\boldsymbol{\xi}}_k}-{\bm\psi}_{{\rm t},n}\rVert}}{\lVert{{\boldsymbol{\xi}}_k}-{\bm\psi}_{{\rm t},n}\rVert}.
\end{align}Here, $k_0=\frac{2\pi}{\lambda}$ is the wavenumber, and $\eta\triangleq\frac{c^2}{16\pi^2f_{\rm{c}}^2}$ with $c$ and $f_{\rm{c}}$ denoting the speed of light and the carrier frequency, respectively. Moreover, the $n$th element of ${\boldsymbol{\phi}}_{\rm t}({\bf x}_{\rm t})$ represents the normalized in-waveguide propagation coefficient between the Tx-PAs and the feed point of the Tx-PASS, which can be written as follows: 
\begin{align}\label{eq:phi_t}
	\left[{\boldsymbol{\phi}}_{\rm t}({\bf x}_{\rm t})\right]_n = \sqrt{\frac{1}{N_{\rm t}}}{\rm exp}\left({\rm -j}k_{\rm g}\lVert{\bm\psi}_{{\rm t},n} - {\bm\psi}_{{\rm t},0}\rVert\right),
\end{align}
where $k_{\rm g} = \frac{2\pi}{\lambda_{\rm g}}$ is the in-waveguide wavenumber with $\lambda_{\rm g}$ denoting the in-waveguide wavelength.
Based on the received signal model in \eqref{eq:Received com signal}, the received SNR at the $k$th user can be written as follows:
\begin{align}\label{eq:SNR for com}
\gamma_k(x_{\rm t})
&\!\triangleq\! \mathbb{E}\!\left(
  \frac{|\mathbf{h}_k^{\mathsf T}(\boldsymbol{\xi}_k,x_{\rm t})\boldsymbol{\phi}_{\rm t}(x_{\rm t})s|^2}{|z_k|^2}
  \right)
\!\!=\! \frac{P_{\rm t}|\mathbf{h}_k^{\mathsf T}(\boldsymbol{\xi}_k,x_{\rm t})\boldsymbol{\phi}_{\rm t}(x_{\rm t})|^2}{\sigma_k^2}.
\end{align}

Accordingly, the achievable MMF rate with a given transmit PA location vector ${\bf x}_{\rm t}$ is determined by the minimum SNR among all multicast users, which can be given by~\cite{9093950}
\begin{align}
{\mathsf R}_{\rm c}({\bf x}_{\rm t}) = \underset{k\in\mathcal{K}}{\mathsf{min}}\left\{\log_{2}\Big(1+\gamma_{k}({\bf x}_{\rm t})\Big)\right\}.\label{eq:rate}
\end{align}
\vspace{-15pt}
\subsection{Sensing Performance Metric}
Next, we consider the PASS-based sensing model, in which the multicast signals transmitted by the Tx-PASS are reflected by the target and subsequently collected by the Rx-PASS for sensing the location of the target. Specifically, the received echo signal is modeled as follows:
\begin{align}
{\bf y}_{{\rm s}} &= \beta_{\rm s}\,{\bf h}({{\boldsymbol{\zeta}}}, {\bf x}_{{\rm r}}){\bf h}^{\mathsf T}({{\boldsymbol{\zeta}}}, {\bf x}_{{\rm t}}){\boldsymbol{\phi}}_{\rm t}({\bf x}_{\rm t}){s} + {\bf z}_{{\rm s}},
\end{align}
where ${\bf z}_{{\rm s}}\sim\mathcal{CN}({\bf 0}, \sigma_{\rm s}^{2}{\bf I}_{N_{\rm r}})$ represents the additive CSCG noise vector, and $\beta_{\rm s}$ denotes the complex target scattering coefficient associated with the radar cross-section (RCS). Here, we treat $\beta_{\rm s}$ as a deterministic constant within one coherent processing interval and absorb it into a normalized gain factor, e.g., $\beta_{\rm s}=1$ without loss of generality~\cite{Hua2024NF3DCRB}, so that the sensing performance is dominated by the geometric dependence on $\boldsymbol{\zeta}$.
Moreover, the vectors ${\bf h}({\boldsymbol{\zeta}}, {\bf x}_{{\rm t}})\in{\mathbb C}^{N_{\rm t}\times1}$ and ${\bf h}({\boldsymbol{\zeta}}, {\bf x}_{{\rm r}})\!\in\!{\mathbb C}^{N_{\rm r}\times1}$ denote the free-space channel vectors between the target and the Tx/Rx-PASS, respectively, and are defined in the same form as those in~\eqref{eq:element_h}.
Denote ${\bf X}\triangleq[{\bf x}_{\rm t}\ {\bf x}_{\rm r}]^{\mathsf T}$ and ${\bf G}({\boldsymbol{\zeta}},{\bf X})\triangleq {\bf h}({\boldsymbol{\zeta}},{\bf x}_{\rm r}){\bf h}^{\mathsf T}({\boldsymbol{\zeta}},{\bf x}_{\rm t})$. 
Then, the corresponding baseband observation at the BS can be expressed as follows:
\begin{align}\label{eq:obs_signal}
\hat{s}\triangleq{\boldsymbol{\phi}}_{\rm r}^{\mathsf T}({\bf x}_{\rm r}){\bf y}_{\rm s}=\beta_{\rm s}\,{\boldsymbol{\phi}}_{\rm r}^{\mathsf T}({\bf x}_{\rm r}){\bf G}({\boldsymbol{\zeta}},{\bf X}){\boldsymbol{\phi}}_{\rm t}({\bf x}_{\rm t})\,{\bf s}+\tilde{z}_{\rm s}.
\end{align}
Here, ${\boldsymbol{\phi}}_{\rm r}({\bf x}_{\rm r})\!\in\!\mathbb{C}^{N_{\rm r}\times1}$ denotes the normalized in-waveguide propagation vector from the Rx-PAs to the feed point of the Rx-PASS, which follows the same modeling principle as ${\boldsymbol{\phi}}_{\rm t}({\bf x}_{\rm t})$ in~\eqref{eq:phi_t}. 
The effective noise after combining is given by $\tilde{z}_{\rm s}={\boldsymbol{\phi}}_{\rm r}^{\mathsf T}({\bf x}_{\rm r}){\bf z}_{\rm s}\sim{\mathcal{CN}}({ 0},\sigma^2_{\rm s})$. Conditioned on the unknown target parameters $\boldsymbol{\zeta}$, $\hat{s}$ is a complex Gaussian distribution as follows:
\begin{align}\label{eq:slotwise_obs}
\hat{s}&\sim\mathcal{CN}\big({\mu}(\boldsymbol{\zeta},{\bf X}),\,\sigma^2_{\rm s}\big),
\end{align}
where
\begin{align}
{\mu}(\boldsymbol{\zeta},{\bf X})&=\beta_{\rm s}\,{\boldsymbol{\phi}}_{\rm r}^{\mathsf T}({\bf x}_{\rm r}){\bf G}({\boldsymbol{\zeta}},{\bf X}){\boldsymbol{\phi}}_{\rm t}({\bf x}_{\rm t}){s}.\label{eq:obs_model_final}
\end{align}

Upon receiving $\hat{s}$, the BS estimates the target locations, yielding an estimate $\hat{\boldsymbol{\zeta}}$. The sensing accuracy is quantified via the mean-square error (MSE), which is defined as follows:
\begin{align}
\mathsf{MSE}_{\boldsymbol{\zeta}}
\triangleq
\mathbb{E}\!\left\{\|\hat{\boldsymbol{\zeta}}-\boldsymbol{\zeta}\|_2^2\right\}.
\label{eq:MSE_def_total}
\end{align}
Since the sensing MSE is difficult to express analytically, we adopt the BCRB as the sensing performance metric, which provides a theoretical lower bound on the MSE~\cite{vantrees2004detection, wang2025hybrid}. Specifically, the BCRB satisfies
\begin{align}
\mathsf{MSE}_{\boldsymbol{\zeta}} \geq \mathsf{BCRB}({\bf X}) =\mathsf{tr}\left\{\mathbf{F}^{-1}\left({\boldsymbol\zeta}, \mathbf X \right) \right\},
\label{eq:BCRB_lower}
\end{align}
where $\mathbf{F}({\boldsymbol\zeta}, \mathbf X)\in\mathbb{R}^{2\times2}$ denotes the Bayesian Fisher information matrix (BFIM) associated with the target location parameter vector $\boldsymbol{\zeta}$. 
It consists of two additive components, namely the observation Fisher information and the prior Fisher information, which satisfies~\cite{5571900}
\begin{align}\label{eq:BFIM}
\mathbf{F}({\boldsymbol\zeta}, \mathbf{X}) =\widetilde{\mathbf{F}}({\boldsymbol\zeta}, \mathbf{X}) +\mathbf{F}^{\rm p}({\boldsymbol\zeta}).
\end{align}
Here, $\widetilde{\mathbf{F}}({\boldsymbol\zeta},\mathbf{X})$ represents the Fisher information contributed by the observation model, and $\mathbf{F}^{\rm p}({\boldsymbol\zeta})$ captures the prior information determined by the statistical distribution of $\boldsymbol{\zeta}$.

\subsubsection{PFIM Derivation}
The prior Fisher information matrix (PFIM) quantifies the prior knowledge regarding ${\boldsymbol{\zeta}}$, which is defined as follows~\cite{kay1993fundamentals}: 
\begin{align}\label{eq:def_F_p}
\mathbf{F}^{\rm p}({\boldsymbol\zeta}) \triangleq -\mathbb{E}_{\boldsymbol\zeta}\!\left[\nabla_{\boldsymbol\zeta}^{2}\ln p(\boldsymbol\zeta)\right].
\end{align}
For analytical tractability, we model the target coordinates as independent Gaussian random variables\footnote{While this work adopts a Gaussian prior for tractability, the framework can be extended to more general distributions, such as Gaussian mixtures, which remains a promising direction for future research.}, i.e., $u^{\alpha} \sim \mathcal{N}(\mu_{{\alpha}}, \sigma_{{\alpha}}^2)$ for ${\alpha} \in \{x, y\}$. Here, the mean $\mu_{\alpha}$ represents the predicted target location derived from historical observations, while the variance $\sigma_{\alpha}^2$ characterizes the associated uncertainty.
Substituting the independent Gaussian PDFs into \eqref{eq:def_F_p}, the log-prior becomes a quadratic function, and the PFIM is derived as the inverse of the prior covariance matrix as follows:
\begin{align}
\mathbf{F}^{\rm p}({\boldsymbol\zeta}) = {\mathsf{Diag}}\big(\sigma^{-2}_{x},\sigma^{-2}_{y} \big).
\end{align}
Note that the off-diagonal elements vanish due to the statistical independence between the $x$ and $y$-coordinates.

\subsubsection{OFIM Derivation}
Similar to the PFIM defined in~\eqref{eq:def_F_p}, the observation Fisher information matrix (OFIM) is defined as follows~\cite{kay1993fundamentals}:
\begin{align}
\widetilde{\mathbf F}(\boldsymbol\zeta, \mathbf X)
&\triangleq
\mathbb E_{\boldsymbol\zeta}\big[
  \nabla_{\boldsymbol\zeta}\ln p( {\hat s}|\boldsymbol\zeta, \mathbf X) \  \  
  \nabla_{\boldsymbol\zeta}\ln p( {\hat s}|\boldsymbol\zeta, \mathbf X)^{\mathsf T}
\big].
\label{eq:def_F_o}
\end{align}
For our CSCG model in~\eqref{eq:slotwise_obs}, the OFIM admits the closed form as follows:
\vspace{-5pt}\begin{align}
\widetilde{\mathbf F}(\boldsymbol\zeta, \mathbf X) &= \frac{2}{\sigma^2_{\rm s}}\mathbb E_{\boldsymbol\zeta}\left[\Re\Big\{\big(\nabla_{\boldsymbol\zeta}\mu\big)^{\mathsf H}\big(\nabla_{\boldsymbol\zeta}\mu\big)\Big\}\right] \notag \\
&\triangleq \frac{2}{\sigma^2}
\mathbb E_{\boldsymbol\zeta}\left[\Re\Big\{\mathbf J^{\mathsf H}(\boldsymbol\zeta, \mathbf X)\mathbf J(\boldsymbol\zeta, \mathbf X)\Big\}\right],\label{eq:Fo_closed}
\end{align}
where
\vspace{-5pt}\begin{align}
\mathbf J(\boldsymbol\zeta, \mathbf X)&=\begin{bmatrix}
\partial_{u^x}\mu \!\!&\!\! \partial_{u^y}\mu
\end{bmatrix} \triangleq
 \begin{bmatrix}
f_{x}(\boldsymbol\zeta, \mathbf X) & f_{y}(\boldsymbol\zeta, \mathbf X)
\end{bmatrix}.
\label{eq:Jrow_def_unified}
\end{align}
By introducing the unified derivative kernel
\begin{align}
    \mathcal K_{\alpha}^{p,n}(\boldsymbol\zeta,\mathbf X)
    \triangleq 
    \eta^{\frac{1}{2}}\,[\boldsymbol\phi_p(\mathbf x_p)]_n
    \frac{{\rm e}^{-{\rm j}k_0 R_{p,n}}\chi_{\alpha}^{p,n}\big(1+{\rm j}k_0 R_{p,n}\big)}
    {R_{p,n}^3},
\end{align}where $\chi_{x}^{p,n}=u^x - x_{p,n}$, $\chi_{y}^{p,n}=u^y$ with $\alpha\in\{x,y\}$, $p\in\{{\rm t},{\rm r}\}$, the functions $f_x$ and $f_y$ in \eqref{eq:Jrow_def_unified} can then be expressed as follows:
\vspace{-5pt}\begin{align}\label{eq:f_alpha}
    f_{\alpha}(\boldsymbol\zeta,\mathbf X)
    = - g_{\rm t}(\boldsymbol\zeta,\mathbf x_{\rm t})
      \sum_{n=1}^{N_{\rm t}} \mathcal K_{\alpha}^{{\rm t},n}
      - g_{\rm r}(\boldsymbol\zeta,\mathbf x_{\rm r})
      \sum_{n=1}^{N_{\rm r}} \mathcal K_{\alpha}^{{\rm r},n},
\end{align}
where $g_{\rm t}(\boldsymbol\zeta,\mathbf x_{\rm t})= \mathbf h^{\mathsf T}(\boldsymbol\zeta,\mathbf x_{\rm t})\,\boldsymbol\phi_{\rm t}(\mathbf x_{\rm t})$ and $g_{\rm r}(\boldsymbol\zeta,\mathbf x_{\rm r}) = \mathbf h^{\mathsf T}(\boldsymbol\zeta,\mathbf x_{\rm r})\,\boldsymbol\phi_{\rm r}(\mathbf x_{\rm r})$
represent the effective channels for the Tx-PASS and Rx-PASS, respectively.

Consequently, the OFIM in~\eqref{eq:Fo_closed} is given by
\begin{align}
\widetilde{\mathbf F}(\boldsymbol\zeta, \mathbf X)  = \frac{2P_{\rm t}}{\sigma^2_{\rm s}}
\begin{bmatrix}
\widetilde F_{xx}(\boldsymbol\zeta, \mathbf X) & \widetilde F_{xy}(\boldsymbol\zeta, \mathbf X)\\[4pt]
\widetilde F_{yx}(\boldsymbol\zeta, \mathbf X) & \widetilde F_{yy}(\boldsymbol\zeta, \mathbf X)
\end{bmatrix},
\label{eq:Fo_block_unified}
\end{align}
where
\begin{align}
\widetilde F_{\alpha\beta}(\boldsymbol\zeta, \mathbf X)\triangleq\mathbb E_{\boldsymbol\zeta}\left[\Re\left\{f_{\alpha}^*(\boldsymbol\zeta, \mathbf X)f_{\beta}(\boldsymbol\zeta, \mathbf X)\right\}\right]\label{eq:Fxx_def}
\end{align}
for $\alpha,\beta\in\{x,y\}$.

For the diagonal elements when $\alpha = \beta$, taking the case of $\alpha=\beta=x$ as an example and substituting the factorized Gaussian prior density $p(\boldsymbol{\zeta})$ into \eqref{eq:Fxx_def} yields
\begin{align}
\widetilde F_{xx}(\boldsymbol\zeta, \mathbf X)\!= \!\Re\!\left\{\!\iint_{\mathbb R^2}\!\!\big|f_{x}(\boldsymbol\zeta, \mathbf X)\big|^2 p(u^x)p(u^y)\,{\rm d}u^x{\rm d}u^y\!\right\}.\label{eq:27a}
\end{align}
By applying the standard variable substitution $u^x=\mu_{x}+\sigma_{x}\nu$ and $u^y=\mu_{y}+ \sigma_{y}\eta$, and noting that $p(u^x){\rm d}u^x \!=\! \frac{e^{-\nu^2}}{\sqrt{\pi}}\!{\rm d}\nu$ and $p(u^y){\rm d}u^y\! = \!\frac{e^{-\eta^2}}{\sqrt{\pi}}\!{\rm d}\eta$, we rewrite the integral in \eqref{eq:27a} into the canonical Gauss-Hermite form as follows:
\begin{multline}
\widetilde F_{xx}(\boldsymbol\zeta, \mathbf X) = \Re\left\{\frac{1}{\pi}\iint_{\mathbb R^2}\big|f_{x}(\mu_{x}+ \sigma_{x}\nu,\big.\right. \\
\left. \big.\mu_{y}+ \sigma_{y}\eta)\big|^2 e^{-\nu^2}e^{-\eta^2}\,{\rm d}\nu\,{\rm d}\eta\right\}. \label{eq:27c}
\end{multline}However, the expectations in \eqref{eq:27c} do not yield closed-form expressions. To address this, we adopt the Gauss-Hermite quadrature (GHQ) rule $\int_{-\infty}^{+\infty}\psi(x)e^{-x^2}{\rm d}x\approx\sum_{i=1}^{T}\omega_i\psi(\xi_i)$ to approximate the integrals, where $\{\omega_i\}$ and $\{\xi_i\}$ denote the weight and abscissa factors of Gauss-Hermite integration, and $T$ is a complexity-vs-accuracy tradeoff parameter. Accordingly, we approximate the FIM entries as follows:
\begin{multline}
\widetilde F_{xx}(\boldsymbol\zeta, \mathbf X)\approx\sum_{i=1}^{T}\sum_{j=1}^{T}\omega_i\omega_j\Re\Big\{\frac{1}{\pi}\big|f_{x}(\mu_{x}+\sqrt{2}\sigma_{x}\xi_i, \\[-10pt]
\mu_{y}+\sqrt{2}\sigma_{y}\xi_j)\big|^2\Big\}.\label{eq:29_mm_xx}
\end{multline}The term $\widetilde F_{yy}(\boldsymbol\zeta, \mathbf X)$ follows an identical derivation by replacing the partial derivative $f_{x}(\cdot)$ with $f_{y}(\cdot)$.

For the off-diagonal entries when $\alpha\neq \beta$, the expectation must be evaluated over the coupled product of the derivatives. Accordingly, the GHQ approximation is applied to the joint term as follows:
\begin{multline}
\widetilde F_{xy}(\boldsymbol\zeta, \mathbf X) \approx \sum_{i=1}^{T}\sum_{j=1}^{T}\omega_i\omega_j\Re\Big\{\frac{1}{\pi}f_{x}^*(\mu_{x}+ \sqrt{2}\sigma_{x}\xi_i, \\
\mu_{y}+ \sqrt{2}\sigma_{y}\xi_j) \cdot f_{y}(\mu_{x}+ \sqrt{2}\sigma_{x}\xi_i, \mu_{y}+ \sqrt{2}\sigma_{y}\xi_j)\Big\}.\label{eq:cross_mn}
\end{multline}Finally, the value of $\widetilde F_{yx}(\boldsymbol\zeta, \mathbf X)$ is obtained via the Hermitian symmetry of the FIM, i.e., $\widetilde F_{yx} = \widetilde F_{xy}$.
Since both $\widetilde{\mathbf{F}}({\boldsymbol\zeta}, \mathbf{X})$ and $\mathbf{F}^{\rm p}({\boldsymbol\zeta})$ are of order $2 \times 2$, the matrix inversion can be calculated in a closed-form, which yields
\begin{align}
\mathsf{BCRB}(\mathbf X)&=\mathsf{tr}\left\{\mathbf F^{-1}(\boldsymbol\zeta, \mathbf X)\right\} \notag\\
&= \frac{\widetilde F_{xx} + \widetilde F_{yy} + \sigma_x^{-2} + \sigma_y^{-2}}{(\widetilde F_{xx} + \sigma_x^{-2})(\widetilde F_{yy} + \sigma_y^{-2}) - |\widetilde F_{xy}|^2}.
\label{eq:BCRB_2x2_exact_SC}
\end{align}

It is clear that both the multicast rate in~\eqref{eq:rate} and the BCRB in~\eqref{eq:BCRB_2x2_exact_SC} depend critically on the PA placement. In general, the PA configuration that is optimal for sensing is not optimal for multicasting, which gives rise to an inherent sensing-communications tradeoff. To characterize and navigate this tradeoff, we consider three complementary design criteria: 1) S-C design, which minimizes the BCRB subject to a multicast-rate requirement; 2) C-C design, which maximizes the MMF multicast rate subject to a sensing-accuracy requirement; and 3) Pareto-optimal design, which characterizes the achievable sensing-communications tradeoff by tracing the Pareto boundary via a rate-profile formulation.
\vspace{-10pt}
\section{Single-PA Case}
In this section, we first investigate a single-PA setting to derive the optimal solution to PA placement for the S-C, C-C, and Pareto-optimal designs.
\vspace{-10pt}
\subsection{Communications-Centric Design}
Since the downlink multicast rate depends only on the Tx-PA placement, the C-C design simplifies to
\begin{align}\label{prob:CCP_single_noSensing}
\max_{x_{\rm t}} \  {\mathsf R}_{\rm c}(x_{\rm t}), \qquad
\text{s.t.}\  x_{\rm t} \in [0,D_{\rm x}].
\end{align}
For the $k$th user located at $(\hat x_{{\rm c},k}, \hat y_{{\rm c},k})$, the received SNR under a single Tx-PA located at $x_{\rm t}$ is
\begin{align}
\gamma_k(x_{\rm t}) = \sigma_k^{-2}{P_{\rm t}}
[{(x_{\rm t}-\hat x_{{\rm c},k})^2 + (y_{\rm t}-\hat y_{{\rm c},k})^2 + h^2}]^{-1}.
\end{align}
Therefore, maximizing the multicast rate is equivalent to minimizing the maximum distance between the Tx-PA and all multicast users, which can be expressed as follows:
\begin{align}
\min_{x_{\rm t}} \; 
\max_{k\in\mathcal K}\bigl[(x_{\rm t}-\hat x_{{\rm c},k})^2 + \Delta^2_k\bigr],
\label{eq:CCP_minmax_noSensing}
\end{align}
where $\Delta^2_k \triangleq (y_{\rm t}-\hat y_{{\rm c},k})^2 + h^2$.

Define $d_k(x)\triangleq (x-\hat x_{{\rm c},k})^2+\Delta^2_k$ and $d(x)\triangleq \max_{k\in\mathcal K}d_k(x)$. 
Each $d_k(x)$ is a convex quadratic function of $x$, and hence $d(x)$ is also convex on $[0,D_{\rm x}]$. 
Since $d_k(x)$ attains its unique minimum at $x=\hat x_{{\rm c},k}$, each $d_k(x)$ is strictly unimodal. 
As $d(x)$ represents the upper envelope of these convex curves, the global minimizer of $d(x)$ must occur at a point where either i) a single curve $d_k(x)$ reaches its minimum, or  ii) two curves intersect and yield equal values. 

Consequently, a closed-form structure of the optimal transmit PA placement can be established by identifying a finite candidate set that guarantees optimality. Specifically, the optimal Tx-PA position can be determined from the following finite candidate set: i) the set of multicast users' $x$-coordinates, $\mathcal X_1 \triangleq \{\hat x_{{\rm c},1},\ldots,\hat x_{{\rm c},K}\}$; and ii) the set of pairwise equal-distance points, $\mathcal X_2 \triangleq \{\xi_{i,j}\,|\, d_i(\xi_{i,j})=d_j(\xi_{i,j}),\, i<j\}$, where each intersection point $\xi_{i,j}$ is given by
\begin{align}
\xi_{i,j}
= \frac{\Delta^2_j - \Delta^2_i}{2(\hat x_{{\rm c},j}-\hat x_{{\rm c},i})}
+ \frac{\hat x_{{\rm c},i} + \hat x_{{\rm c},j}}{2}.
\end{align}
Including the boundary points, the complete candidate set is $\mathcal G_{\rm c} = \{0, D_{\rm x}\} \cup \mathcal X_1 \cup \mathcal X_2$. Finally, the globally optimal Tx-PA placement for the C-C design is obtained by evaluating the multicast objective at all candidate points in $\mathcal G_{\rm c}$ as follows:
\begin{align}
x_{\rm t}^\star 
= \arg\min_{x\in \mathcal G_{\rm c}} 
\max_{k\in\mathcal K}\bigl[(x-\hat x_{{\rm c},k})^2 + \Delta^2_k\bigr].
\end{align}
\vspace{-15pt}
\subsection{Sensing-Centric Design}\label{S_C_singlePA}
Obtaining closed-form expressions for the optimal Tx- and Rx-PA placements that minimize the BCRB in~\eqref{eq:BCRB_2x2_exact_SC} is in general intractable, because the sensing performance depends on the joint Tx- and Rx-PA placement. In particular, the off-diagonal terms of the BFIM couple the two variables, so that the optimal Tx-PA location depends on the Rx-PA location and vice versa. Consequently, the first-order optimality conditions lead to coupled nonlinear equations in $(x_{\rm t},x_{\rm r})$, for which a closed-form solution is generally unavailable.

To expose the key geometric structure of the S-C design, we therefore adopt a simplified symmetric setup in which the target lies on the perpendicular bisector of the two waveguides. Specifically, we treat its $y$-coordinate as known and fixed at $u^y=\frac{y_{\rm t}+y_{\rm r}}{2}$.
Under the single-PA configuration, the Euclidean distances between the target and the Tx- or Rx-PAs are expressed as follows:
\begin{align}
R_{\rm t}=\sqrt{(u^x-x_{\rm t})^2+\Delta_{\rm s}^2}, \ 
R_{\rm r}=\sqrt{(u^x-x_{\rm r})^2+\Delta_{\rm s}^2},
\end{align}
 where $\Delta_{\rm s}=\sqrt{h^2+\big(\tfrac{y_{\rm r}-y_{\rm t}}{2}\big)^2}$. Consequently, we focus the estimation solely on the $x$-coordinate $u^x$. Based on this assumption, the mean of the received echo signal in \eqref{eq:obs_model_final} simplifies to the following:
\vspace{-10pt} \begin{align}
 \mu\left(u^x, x_{\rm t}, x_{\rm r}\right)=\frac{C {\rm e}^{-{\rm j} k_0\left(R_{\mathrm{t}}+R_{\mathrm{r}}\right)}}{R_{\mathrm{t}} R_{\mathrm{r}}},
 \end{align}
 where $C \triangleq \beta_{\rm s}\sqrt{ \eta P_{\mathrm{t}}} {\rm e}^{-{\rm j} k_{\mathrm{g}}\left(x_{\mathrm{t}}+x_{\mathrm{r}}\right)}$ aggregates the system scaling factors and the phase shift introduced by in-waveguide propagation. 
 Since the unknown parameter reduces to the scalar $u^x$, the 
\emph{conditional} observation Fisher information for estimating $u^x$ is a scalar quantity, denoted by $F_{xx}(u^x,x_{\rm t},x_{\rm r})$.
Taking the prior expectation over $u^x$, the corresponding scalar OFIM entry in~\eqref{eq:Fo_block_unified} can be simplified as follows:
\begin{align}
\widetilde{\mathbf F}(u^x, \mathbf X) = \widetilde F_{xx}(u^x, x_{\rm t},x_{\rm r}) \triangleq \mathbb E_{u^x}\big[F_{xx}(u^x,x_{\rm t},x_{\rm r})\big].
\end{align}
Accordingly, the one-dimensional BFIM becomes $F(x_{\rm t},x_{\rm r}) = \widetilde F_{xx}(u^x, x_{\rm t},x_{\rm r}) + \sigma_x^{-2}$,
and the BCRB in~\eqref{eq:BCRB_2x2_exact_SC} simplifies to
\begin{align}
\mathsf{BCRB}(u^x, x_{\rm t},x_{\rm r}) = \left({\widetilde F_{xx}(u^x, x_{\rm t},x_{\rm r})+\sigma_x^{-2}}\right)^{-1}.
\end{align}
For a fixed prior term $\sigma_x^{-2}$, improving the sensing accuracy is thus equivalent to increasing $\widetilde F_{xx}(u^x, x_{\rm t},x_{\rm r})$. Consequently, we formulate the S-C design problem as follows:
\begin{align}\label{eq:1D_BCRB_obj}
\max_{\{x_{\rm t}, x_{\rm r}\}} \widetilde F_{xx}(u^x, x_{\rm t},x_{\rm r}),\  \text{s.t.}\ x_{\rm t} \in [0,D_{\rm x}],   x_{\rm r} \in [0,D_{\rm x}].
\end{align}
To make the dependence of the sensing information on the PA placement explicit, we next rewrite the conditional Fisher information in terms of the LoS distances and angles.
Specifically, define the direction cosine of the Tx/Rx links along the $x$-axis as
${\rm cos}\, \theta_{p} \triangleq (u^x - x_{p})/R_p$ for ${p} \in \{{\rm t}, {\rm r}\}$.
Then, the conditional Fisher information for estimating $u^x$ can be expressed as follows:
\begin{align}
F_{xx}(\theta_{\rm t}, \theta_{\rm r})\!& =\! \frac{2|\mu|^2}{\sigma_{\rm s}^2}  \Big[ \left(k_0^2 + R_{\rm t}^{-2}\right){\rm cos}\,\theta_{\rm t}^2 + \left(k_0^2 + R_{\rm r}^{-2}\right){\rm cos}\,\theta_{\rm r}^2 \notag \\ 
&+ 2\left(k_0^2 + (R_{\rm t}R_{\rm r})^{-1}\right){\rm cos}\,\theta_{\rm t}{\rm cos}\,\theta_{\rm r}\Phi \Big],
\label{eq:Fxx_exact_compact}
\end{align}
where $\Phi \triangleq \cos\big(k_0(R_{\rm t}-R_{\rm r})\big)$, and $|\mu|^2$ denotes the path-loss dependent signal power.

Directly solving~\eqref{eq:1D_BCRB_obj} is still challenging, because it involves a prior expectation and the conditional term $F_{xx}(u^x,x_{\rm t},x_{\rm r})$ is highly non-convex with respect to (w.r.t.) the transmit and receive PA placements.
To obtain a tractable design with clear physical insight, we proceed in two steps.
First, for a fixed target realization $u^x$, we study the conditional Fisher information $F_{xx}(u^x,x_{\rm t},x_{\rm r})$ and show that it is maximized when the Tx-PA and Rx-PA are placed symmetrically around $u^x$ along the $x$-axis (see Proposition~1). This restricts the transceiver layout to a symmetric manifold that can be parameterized by a center $c$ and a nonnegative displacement $d$.
Second, under a symmetric prior for $u^x$, we show that the Bayes-optimal center is $c=\mu_x$ (see Proposition~2), which leads to a closed-form displacement summarized in Remark~1.
\vspace{-5pt}
\begin{proposition}
For any given target $x$-coordinate $u^x$, under the symmetric setup $u^y=(y_{\rm t}+y_{\rm r})/2$, 
the conditional observation Fisher information $F_{xx}(u^x,x_{\rm t},x_{\rm r})$ is maximized when the Tx-PA and Rx-PA are placed symmetrically around the target along the $x$-axis, i.e., $|u^x-x_{\rm t}|=|u^x-x_{\rm r}|$.
\end{proposition}
\vspace{-5pt}
\begin{IEEEproof}
Please refer to Appendix~A.
\end{IEEEproof}
Motivated by Proposition~1, we parameterize this symmetric design class by a \emph{symmetry center $c$} and a nonnegative \emph{$x$-axis displacement $d\ge 0$} (i.e., $d$ is the absolute $x$-coordinate difference from the center), namely
\begin{align}
\mathcal M(c)\triangleq 
\Big\{(x_{\rm t},x_{\rm r}) \,\big|\, x_{\rm t}=c\pm d,\; x_{\rm r}=c\pm d,\; d\ge 0\Big\}.
\label{eq:manifold_c}
\end{align}
Restricting $(x_{\rm t},x_{\rm r})$ to $\mathcal M(c)$ reduces the two-dimensional placement to a one-dimensional search over $d$ for any fixed $c$.
We select the same-side branch $x_{\rm t}=x_{\rm r}=c-d$. In contrast, for the opposite-side equal-offset placement $x_{\rm t}=c-d$ and $x_{\rm r}=c+d$, the geometric placement implies $R_{\rm t}=R_{\rm r}$ and $\cos\theta_{\rm r}=-\cos\theta_{\rm t}$, Substituting these relations into \eqref{eq:Fxx_exact_compact} shows that the bracketed term cancels, and hence $F_{xx}=0$. 
Moreover, under a symmetric prior on $\mu_x$, the Bayes-aligned center is attained at the prior mean $c=\mu_x$, as established in Proposition~2.
\vspace{-5pt}
\begin{proposition}\label{prop:center_and_disp}
Assume that the prior of $u^x$ is symmetric about $\mu_x$. Consider the OFIM restricted to the equal-offset manifold $\mathcal M(c)$, i.e., $\widetilde F_{xx}(c,d)\triangleq \mathbb E_{u^x}\!\big[F_{xx}(u^x,c,d)\big]$. The OFIM is stationary w.r.t. the center at the prior mean, i.e., $\partial_c\,\widetilde F_{xx}(c,d)\big|_{c=\mu_x}=0$, which motivates choosing $c=\mu_x$ as the canonical Bayes-aligned center under a symmetric prior. After that, the derivative of the OFIM w.r.t. $d$ admits
\begin{align}\label{eq:derive_Fxx}
\frac{{\rm d}}{{\rm d}d}\,\mathbb{E}_{u^x}\!\big[ F_{xx}(u^x,d)\big]
= \partial_d F_{xx}(\mu_x,d)+\mathcal O(\sigma_x^2),
\end{align}
\end{proposition}
\begin{IEEEproof}
Please refer to Appendix~B.
\end{IEEEproof}

By setting the first derivative of~\eqref{eq:derive_Fxx} to zero, we obtain the optimal $x$-axis displacement $d^\star$, summarized in the following remark.
\vspace{-5pt}
\begin{remark}
The closed-form expression of the optimal displacement is
\vspace{-5pt}
\begin{align}
(d^\star)^2
= \frac{-(k_0^{2}\Delta_{\rm s}^{2}+3)
+ \sqrt{9k_0^{4}\Delta_{\rm s}^{4}+14k_0^{2}\Delta_{\rm s}^{2}+9}}
{4k_0^{2}}.
\label{eq:closed_form_d}
\end{align}
In the typical high-frequency regime of PASS where $k_0\Delta_{\rm s}\gg 1$, it further admits the approximation $d^\star \approx \Delta_{\rm s}/\sqrt{2}$.
\end{remark}
\vspace{-5pt}
\begin{IEEEproof}
Please refer to Appendix~C.
\end{IEEEproof}

%
By additionally incorporating the in-waveguide propagation loss, the resulting closed-form PA placement becomes
\begin{equation}
x_{\rm t}^\star = x_{\rm r}^\star = \mu_x - d^\star.
\end{equation}
This result highlights a non-intuitive but important design insight: an intuition is to place the PA as close to the target as possible, i.e., $x_{\rm t}=x_{\rm r}=\mu_x$ (equivalently, $d\to 0$). However, \eqref{eq:closed_form_d} shows that this intuition is generally incorrect: the optimal displacement $d^{\star}$ is a function of the effective height $\Delta_{\rm s}$, which is consistent with the observation in~\cite{ding2025pinchingisac}.
\vspace{-10pt}
\subsection{Pareto-Optimal Design}\label{sec:Pareto_sPA}
Building on the preceding C-C and S-C designs, we now characterize the fundamental communications-sensing tradeoff via the rate-profile scalarization framework. Specifically,
we define the \emph{sensing rate} as the inverse of the BCRB, i.e., $\mathsf{R}_{\rm s}(x_{\rm t},x_{\rm r})\triangleq\big[\mathsf{BCRB}(x_{\rm t},x_{\rm r})\big]^{-1}$.
For a given $\alpha$, a Pareto-optimal point is obtained by solving
\begin{align}\label{eq:PS_rate_profile_master_single}
\max_{x_{\rm t},x_{\rm r},\mathsf R}~ \mathsf R,
\quad \text{s.t.}~ {\mathsf R}_{\rm c}(x_{\rm t}) \ge \alpha\mathsf R,
\ \mathsf R_{\rm s}(x_{\rm t},x_{\rm r}) \ge (1-\alpha)\mathsf R.
\end{align}
Equivalently, introducing $\mathsf R$ yields the max-min form
\begin{align}\label{eq:slot_value_single}
\max_{x_{\rm t},x_{\rm r}}
\min\!\left\{
\frac{{\mathsf R}_{\rm c}(x_{\rm t})}{\alpha+\delta(\alpha,0)}, \ 
\frac{\mathsf R_{\rm s}(x_{\rm t},x_{\rm r})}{(1-\alpha)+\delta(\alpha,1)}
\right\},
\end{align}
where $\delta(\cdot)$ is a small perturbation used only to avoid numerical issues at $\alpha\in\{0,1\}$.
In the single-PA setting, $\mathsf R_{\rm c}(x_{\rm t})$ is independent of the Rx-PA location $x_{\rm r}$, which enables a simple two-step solution: (i) compute a sensing-driven best response for $x_{\rm r}$ given $x_{\rm t}$; (ii) perform a univariate search over $x_{\rm t}$.

\subsubsection{Sensing-Driven Rx-PA Placement}
In \eqref{eq:slot_value_single}, $x_{\rm r}$ appears only in the sensing term. Hence, for any fixed $x_{\rm t}$, the Rx-PA should be placed to maximize $\mathsf R_{\rm s}(x_{\rm t},x_{\rm r})$, regardless of $\alpha$:
\begin{align}\label{eq:rx_best_response}
 x_{\rm r}^\star(x_{\rm t})
 = \argmax_{x_{\rm r}}\, \mathsf R_{\rm s}(x_{\rm t},x_{\rm r})
 = \argmin_{x_{\rm r}}\, \mathsf{BCRB}(x_{\rm t},x_{\rm r}).
\end{align}
Here, we continue to adopt the sensing model used in the S-C design, where the target's $y$-coordinate is assumed known and only its $x$-coordinate is uncertain. Under this setting, the sensing-centric analysis yields an symmetric layout w.r.t. the prior mean $\mu_x$ as follows:
\begin{align}
 x_{\rm r}^\star(x_{\rm t}) \approx 2\mu_x - x_{\rm t}.
\end{align}

\subsubsection{Tx-PA Placement Optimization}
Substituting $x_{\rm r}^\star(x_{\rm t})$ into \eqref{eq:slot_value_single} reduces the original bivariate problem to a univariate optimization over $x_{\rm t}$.
Although $x_{\rm t}$ is continuous in principle, this one-dimensional search can be efficiently implemented via discretization.
Specifically, we discretize the feasible interval $[0,D_{\rm x}]$ into an $L$-point grid as follows:
\begin{align}\label{eq:waveguide_dis}
	\mathcal{G} \triangleq \Big\{0,\, \tfrac{D_{\rm x}}{L-1},\, \tfrac{2D_{\rm x}}{L-1},\,\ldots,\, D_{\rm x}\Big\}.
\end{align}
Then, for a given $\alpha$, a near-optimal Tx-PA placement is selected by evaluating the objective over $\mathcal{G}$ as follows:
\begin{align}\label{eq:tx_outer_search}
 x_{\rm t}^\star(\alpha)
 = \argmax_{x_{\rm t}\in\mathcal{G}}
 \min\!\left\{
 \frac{{\mathsf R}_{\rm c}(x_{\rm t})}{\alpha+\delta(\alpha,0)},
 \frac{\mathsf R_{\rm s}\big(x_{\rm t},x_{\rm r}^\star(x_{\rm t})\big)}{(1-\alpha)+\delta(\alpha,1)}
 \right\}.
\end{align}
\section{Multi-PA Case}
We further extend our study to the multi-PA setting to enhance practical relevance by explicitly incorporating additional performance constraints beyond the objective function~\cite{10217169, 10639496}. 
\vspace{-15pt}
\subsection{Sensing-Centric Design}
For the S-C design, the PA placement is optimized to enhance the sensing performance for the target, while ensuring a minimum multicast rate. The corresponding optimization problem is formulated as follows:
\begin{align}\label{eq:S_C_obj}
\min_{\mathbf X} \mathsf{BCRB}(\mathbf X), \
\text{s.t.}\ {\mathsf R}_{\rm c}\big({\bf x}_{\rm t}\big)\geq \Gamma_{\rm c},\ \mathbf x_{\rm t}\in\mathcal X_{\rm t},\ \mathbf x_{\rm r}\in\mathcal X_{\rm r}.
\end{align}
Obtaining the globally optimal solution would require an exhaustive search over all feasible PA placement, which is computationally prohibitive. To address this challenge, we develop a sequential element-wise AO algorithm, where each PA location is updated in turn while keeping the others fixed.

\subsubsection{Element-wise BFIM Reformulation}
Fix all PA coordinates except $x_{p,q}$.
For $\alpha\in\{x,y\}$, define
$S_{p,\alpha}\triangleq\sum_{n=1}^{N_p}\mathcal K_{\alpha}^{p,n}$ and
$g_p\triangleq \mathbf h^{\mathsf T}(\boldsymbol\zeta,\mathbf x_p)\boldsymbol\phi_p(\mathbf x_p)$.
When updating $x_{p,q}$, we split
\begin{align}
g_p = g_p^{(-q)} + g_p^{(q)}(x_{p,q}),\quad
S_{p,\alpha}=S_{p,\alpha}^{(-q)}+\mathcal K_{\alpha}^{p,q}(x_{p,q}),
\end{align}
where
$g_p^{(-q)}\triangleq\sum_{n\neq q} h_{p,n}(\boldsymbol\zeta;x_{p,n})[\boldsymbol\phi_p]_n(x_{p,n})$
and
$g_p^{(q)}(x_{p,q})\triangleq h_{p,q}(\boldsymbol\zeta;x_{p,q})[\boldsymbol\phi_p]_q(x_{p,q})$.

Accordingly, $f_{\alpha}(\boldsymbol\zeta,\mathbf X)$ in~\eqref{eq:f_alpha} admits the exact element-wise decomposition as follows:
\begin{align}\label{eq:split_fx}
f_\alpha(\boldsymbol\zeta,\mathbf X)
= C_{\alpha}^{-(p,q)}(\boldsymbol\zeta,\mathbf X_{p,q}^{-})
+ A_{\alpha}^{(p,q)}(\boldsymbol\zeta,x_{p,q}),
\end{align}
with
\begin{subequations}
\begin{align}
C_{\alpha}^{-(p,q)}
&\triangleq
- g_p^{(-q)} S_{p,\alpha}^{(-q)}
- g_{\bar p} S_{\bar p,\alpha},\\
A_{\alpha}^{(p,q)}
&\triangleq
- g_p^{(-q)} \mathcal K_{\alpha}^{p,q}
- g_p^{(q)} S_{p,\alpha}^{(-q)}
- g_p^{(q)} \mathcal K_{\alpha}^{p,q},
\end{align}
\end{subequations}
where $\bar p$ denotes the other waveguide index (i.e., $\bar{\rm t}={\rm r}$ and $\bar{\rm r}={\rm t}$). Note that $g_{\bar p}$ and $S_{\bar p,\alpha}$ depend only on $\mathbf x_{\bar p}$, hence they are constant when optimizing $x_{p,q}$.
By construction, $C_{\alpha}^{-(p,q)}$ is independent of $x_{p,q}$ and all the dependence on $x_{p,q}$ is captured by $A_{\alpha}^{(p,q)}$.
Substituting~\eqref{eq:split_fx} into~\eqref{eq:Fxx_def}, each OFIM entry admits an explicit element-wise dependence on the single variable $x_{p,q}$ through $A_{\alpha}^{(p,q)}(\boldsymbol\zeta,x_{p,q})$. After taking expectation w.r.t. $\boldsymbol\zeta$, the $(\alpha,\beta)$th OFIM entry can be written as follows:
\begin{align}
\big[\widetilde{\mathbf F}(x_{p,q})\big]_{\alpha,\beta}
&= \mathbb E_{\boldsymbol\zeta}\!\left[\Re\!\left\{f_{\alpha}^* f_{\beta}\right\}\right] \notag\\
&= [\mathbf \Phi]_{\alpha,\beta} + [\mathbf \Lambda(x_{p,q})]_{\alpha,\beta} + [\mathbf \Omega(x_{p,q})]_{\alpha,\beta},
\end{align}
where
\begin{subequations}
\begin{align}
[\mathbf \Phi]_{\alpha,\beta} &= \mathbb E\!\left[\Re\!\left\{(C_{\alpha}^{-(p,q)})^{*} C_{\beta}^{-(p,q)}\right\}\right],\\
[\mathbf \Lambda(x_{p,q})]_{\alpha,\beta} &= 2\Re\!\left\{\mathbb E\!\left[(C_{\alpha}^{-(p,q)})^{*} A_{\beta}^{(p,q)}\right]\right\},\\
[\mathbf \Omega(x_{p,q})]_{\alpha,\beta} &= \mathbb E\!\left[(A_{\alpha}^{(p,q)})^{*} A_{\beta}^{(p,q)}\right].
\end{align}
\end{subequations}
Accordingly, the BFIM in~\eqref{eq:BFIM} can be expressed as a function of $x_{p,q}$ as follows:
\begin{align}\label{eq:Fo_mm_elementwise}
\mathbf F(x_{p,q}) = \mathbf \Phi + \mathbf \Lambda(x_{p,q}) + \mathbf \Omega(x_{p,q}) + \mathbf F^{\rm p}_{\boldsymbol\zeta}.
\end{align}

\subsubsection{Element-wise Multicast Rate Reformulation}
We next derive an element-wise expression for the multicast rate.
When updating the $q$th Tx-PA location $x_{{\rm t},q}$ and keeping the remaining Tx-PAs fixed, the channel vector for the $k$th user can be decomposed as follows:
\begin{align}
{\bf h}_{k}^{\mathsf T}\big(\boldsymbol{\xi}_k,{\bf x}_{\rm t}\big)
= {\bf h}_{k}^{(-q)}\big(\boldsymbol{\xi}_k,{\bf x}^{-}_{{\rm t}, q}\big)
+ a_{k}(x_{{\rm t}, q}),
\end{align}
where ${\bf h}_{k}^{(-q)}$ collects the contributions from all fixed Tx-PAs $\{x_{{\rm t},j}\}_{j\neq q}$ and $a_k(x_{{\rm t},q})$ denotes the contribution of the $q$th element.
Let $D_k(x_{{\rm t}, q})\triangleq\lVert{\boldsymbol{\xi}}_k-{\bm\psi}_{{\rm t},q}\rVert$ be the distance between the $q$th Tx-PA and the $k$th user. Then,
\begin{align}
{\bf h}_{k}^{(-q)}\big(\boldsymbol{\xi}_k,{\bf x}_{\rm t}\big)
&= \sum_{j \neq q}^{N_{\rm t}}
\frac{\sqrt{\eta}{\rm e}^{-\,{\rm j}\,\left(k_0 D_{k}(x_{{\rm t}, j})+k_{\rm g} x_{{\rm t}, j}\right)}}{D_{k}(x_{{\rm t}, j})},
\end{align}
and
\vspace{-8pt}
\begin{align}
a_{k}(x_{{\rm t}, q})
&= \frac{\sqrt{\eta}{\rm e}^{-{\rm j}\left(k_0D_{k}(x_{{\rm t}, q})+k_{\rm g}x_{{\rm t}, q}\right)}}{D_{k}(x_{{\rm t}, q})}.
\end{align}
Consequently, the effective channel gain admits the expansion
\begin{multline}
\left|
{\bf h}_{k}^{\mathsf T}\big(\boldsymbol{\xi}_k,{\bf x}_{\rm t}\big)\boldsymbol{\phi}_{\rm t}\big({\bf x}_{\rm t}\big)
\right|^{2}
= \left|{\bf h}_{k}^{(-q)}\boldsymbol{\phi}_{\rm t}\right|^{2}
+ \left|a_{k}(x_{{\rm t}, q})\boldsymbol{\phi}_{\rm t}\right|^{2} \\
+2\Re\!\left\{
\left({\bf h}_{k}^{(-q)}\boldsymbol{\phi}_{\rm t}\right)^{*}
\left(a_{k}(x_{{\rm t}, q})\boldsymbol{\phi}_{\rm t}\right)
\right\}.
\label{eq:h_expand_my}
\end{multline}
Substituting~\eqref{eq:h_expand_my} into~\eqref{eq:SNR for com}, the received SNR of user $k$ becomes an explicit function of $x_{{\rm t},q}$, i.e., 
\begin{align}\label{eq:gamma_element}
\gamma_k\big(x_{{\rm t}, q}\big)
= \frac{P_{\rm t}}{\sigma_k^2}
\Big(C_k + Q_k\big(x_{{\rm t}, q}\big) + L_k\big(x_{{\rm t}, q}\big)\Big),
\end{align}
where
\vspace{-8pt}
\begin{align}
&C_k = \left|{\bf h}_{k}^{(-q)} \boldsymbol{\phi}_{\rm t}\right|^{2}, \quad
Q_k\big(x_{{\rm t}, q}\big) = \left|a_{k}\big(x_{{\rm t}, q}\big)\boldsymbol{\phi}_{\rm t}\right|^{2}, \notag\\
&L_k\big(x_{{\rm t}, q}\big) = 2\Re\!\left\{
\left({\bf h}_{k}^{(-q)}\boldsymbol{\phi}_{\rm t}\right)^{*}
\left(a_{k}\big(x_{{\rm t}, q}\big)\boldsymbol{\phi}_{\rm t}\right)
\right\}.
\end{align}
\subsubsection{Sequential Element-Wise AO Procedure}
We now explain how the element-wise BFIM reformulation in \eqref{eq:Fo_mm_elementwise} and the element-wise SNR reformulation in \eqref{eq:gamma_element} enable a sequential element-wise AO procedure.
Following the discretization process in Section~\ref{sec:Pareto_sPA}, we first quantize the feasible interval $[0,D_{\rm x}]$ into an $L$-point grid as~\eqref{eq:waveguide_dis}, so that each element-wise update can be solved by a one-dimensional grid search.

At a given iteration, we update one PA coordinate $x_{p,q}$ while keeping all other PA coordinates fixed.
To preserve the minimum-spacing constraint after each update, the candidate set of $x_{p,q}$ is restricted to a local feasible set as follows:
\begin{align}
\mathcal{S}_{p,q}\!\!\triangleq\!
\Big\{\!x\in\mathcal{G}\,\big|\,|x-x_{p,j}|\ge \Delta_{\min},\forall j\in\{q\!-\!1,q\!+\!1\}\Big\},
\end{align}
where $x_{p,q-1}$ and $x_{p,q+1}$ are the fixed neighboring PAs on the same waveguide. For any candidate point $x\in\mathcal{S}_{p,q}$, the BFIM can be evaluated efficiently using the element-wise form \eqref{eq:Fo_mm_elementwise}, i.e., by treating $\mathbf F(x)$ as a function of the single scalar variable $x_{p,q}=x$ and keeping the remaining coordinates fixed.
Similarly, when $p=\mathrm{t}$, the multicast SNRs can be evaluated via \eqref{eq:gamma_element} for all users $k\in\mathcal{K}$ as explicit functions of the same scalar variable $x$.
\paragraph*{Tx-PA update}
Updating a Tx-PA affects both sensing and communications. Hence, the $q$th Tx-PA is updated by solving
\begin{align}\label{eq:sub_problem_Tx}
x_{{\rm t},q}^{\star}\!=\!\argmin_{x\in\mathcal{S}_{{\rm t},q}}\mathsf{tr}\!\left\{\mathbf F^{-1}(x)\right\}, \
\text{s.t. } \gamma_k(x)\ge \gamma_{\rm c},\forall k\in\mathcal{K},
\end{align}
where $\gamma_{\rm c}\triangleq 2^{\Gamma_{\rm c}}-1$.
Notably, the multicast constraint is enforced {during} the one-dimensional search: we only admit candidate points $x$ that satisfy $\gamma_k(x)\ge \gamma_{\rm c}$ for all $k$.
Therefore, if the initial Tx-PA layout $\mathbf x_{\rm t}^{(0)}$ is feasible, i.e., ${\mathsf R}_{\rm c}(\mathbf x_{\rm t}^{(0)})\ge\Gamma_{\rm c}$, then every subsequent Tx-PA update keeps the iterate feasible, and the multicast requirement is guaranteed throughout the AO iterations.
\begin{algorithm}[t]
\caption{Element-wise AO Algorithm for the S-C Design}
\label{alg:elementwise}
\begin{algorithmic}[1]
\REQUIRE  Initial layout $\mathbf X^{(0)}$ and set $l\gets 0$.
\REPEAT
    \FOR{$q\in\mathcal N_{\rm t}$}
        \STATE Construct $\mathcal S_{{\rm t},q}$ and update $x_{{\rm t},q}$ by solving~\eqref{eq:sub_problem_Tx}.
    \ENDFOR
    \FOR{$q\in\mathcal N_{\rm r}$}
        \STATE Construct $\mathcal S_{{\rm r},q}$ and update $x_{{\rm r},q}$ by solving~\eqref{eq:sub_problem_Rx}.
    \ENDFOR
    \STATE $l\gets l+1$.
\UNTIL{convergence or $l=I_{\rm iter}$}
\STATE \textbf{Output:} $\mathbf X^{(l)}$.
\end{algorithmic}
\end{algorithm}
\paragraph*{Rx-PA update}
Updating an Rx-PA does not affect multicast communications. Thus, the $q$th Rx-PA is updated by solving the unconstrained one-dimensional search as follows:
\begin{align}\label{eq:sub_problem_Rx}
x_{{\rm r},q}^{\star}=\argmin_{x\in\mathcal{S}_{{\rm r},q}}~\mathsf{tr}\!\left\{\mathbf F^{-1}(x)\right\}.
\end{align}

By cyclically sweeping $q$ over all Tx- and Rx-PAs and performing the above one-dimensional searches, we obtain a sequential element-wise AO algorithm, as detailed in Algorithm 1.
Since each subproblem is solved over a finite set and each accepted update does not increase the objective, the procedure converges in a finite number of iterations.
Moreover, the computational complexity is $\mathcal{O}\big(I_{\rm iter}(N_{\rm t}+N_{\rm r})\,L\,(K+T^{2})\big)$, where $I_{\rm iter}$ denotes the number of AO iteration.
\begin{algorithm}[t]
\caption{AL-Based AO Algorithm for the C-C Design}
\label{alg:PS_AL_short}
\begin{algorithmic}[1]
\REQUIRE Initial layout $\mathbf X^{(0)}$, multiplier $\lambda^{(0)}\!\ge0$, 
penalty $\rho^{(0)}\!>\!0$, growth factor $\beta\!>\!1$, 
tolerances $\varepsilon_{\rm in}, \varepsilon_{\rm out}$, set $s\gets0$, $l \gets 0$.
\REPEAT 
  \REPEAT  
    \STATE Update $\mathbf X$ by solving~\eqref{eq:tx_scalar_update} and~\eqref{eq:rx_scalar_update} via Algorithm 1
    \STATE Update Lagrange multiplier 
           $\lambda^{(l+1)}$ by~\eqref{eq:PS_lambda_update_only}
    \STATE $l \gets l+1$
  \UNTIL{convergence \textbf{or} $l=I_{\max}$}
  \STATE Update penalty parameter via~\eqref{eq:penalty_update}
  \STATE $s \gets s+1$

\UNTIL{$|\mathsf R_{\rm c}(\mathbf x_{\rm t}^{(l)}) - 
         \mathsf R_{\rm c}(\mathbf x_{\rm t}^{(l-1)})|
         \le \varepsilon_{\rm out}$}
\STATE \textbf{Output:} $\mathbf X^{(l)}$.
\end{algorithmic}
\end{algorithm}
\vspace{-10pt}
\subsection{Communications-Centric Design}
For the C-C design, the PA placement is optimized to maximize the multicast rate, subject to maintaining the sensing accuracy above a desired threshold. This leads to the following optimization formulation:
\begin{align}\label{eq:CC_PS_only}
\max_{{\bf x}_{\rm t}}\ \mathsf R_{\rm c}\big({\bf x}_{\rm t}\big), \
{\rm s.t.}\ \mathsf{BCRB}\big({\bf X}\big)\le\Gamma_{\rm s},  \mathbf x_{\rm t}\in\mathcal X_{\rm t}, \mathbf x_{\rm r}\in\mathcal X_{\rm r}.
\end{align}
A key difficulty in~\eqref{eq:CC_PS_only} is that the receive-side variables $\mathbf x_{\rm r}$ do not appear in the objective $\mathsf R_{\rm c}(\mathbf x_{\rm t})$. Therefore, directly applying the element-wise AO method from the S-C design would leave $\mathbf x_{\rm r}$ unchanged.
To enable joint transceiver updates, we adopt an augmented Lagrangian (AL) reformulation, which embeds the sensing constraint into the objective and thus induces an explicit dependence on $\mathbf x_{\rm r}$.
As a result, the same element-wise update mechanism developed in the S-C design can be reused as the inner AO step.

Specifically, define the constraint violation as $\Delta_{\rm B}({\bf X})\triangleq\mathsf{BCRB}({\bf X})-\Gamma_{\rm s}$.
Let $\lambda\ge0$ denote the Lagrange multiplier and $\rho>0$ denote the penalty parameter.
The AL objective is
\begin{align}\label{eq:AL_PS_only}
\mathcal L_\rho\big({\bf X},\lambda\big)
\triangleq \mathsf R_{\rm c}\big({\bf x}_{\rm t}\big)
- \lambda\,\Delta_{\rm B}\big({\bf X}\big)
- \frac{\rho}{2}\Big[\Delta_{\rm B}\big({\bf X}\big)\Big]_+^{2}.
\end{align}
With~\eqref{eq:AL_PS_only}, we alternate between (i) updating $\mathbf X$ via sequential element-wise AO and (ii) updating $(\lambda,\rho)$ to enforce feasibility of the sensing constraint.
\subsubsection{Optimization of $\mathbf{X}$}
At outer iteration $l$, given the current multiplier $\lambda^{(l)}$ and penalty parameter $\rho^{(l)}$, the PA layout $\mathbf X$ is updated by alternating between ${\bf x}_{\rm t}$ and ${\bf x}_{\rm r}$.

For the transmit side, when optimizing a specific coordinate $x_{{\rm t},q}$ subject to its feasible set $\mathcal S_{{\rm t},q}$, the corresponding univariate objective is defined as follows:
\begin{align}
\phi^{(l)}_{{\rm t},q}(x_{{\rm t}, q})\!\triangleq \!
\mathsf R_{\rm c}(x_{{\rm t}, q})\!- \!\lambda^{(l)}\!\Delta_{\rm B}(x_{{\rm t}, q}) \!- \!\frac{\rho^{(l)}}{2}\Big[\Delta_{\rm B}(x_{{\rm t}, q})\Big]_+^2,
\label{eq:tx_scalar_obj}
\vspace{-15pt}
\end{align}
and the $q$th Tx-PA position is then updated by solving
\begin{align}
x_{{\rm t},q}^{(l+1)} 
= \argmax_{x_{{\rm t},q}\in\mathcal S_{{\rm t},q}}
\phi^{(l)}_{{\rm t},q}(x_{{\rm t}, q}),
\qquad q\in \mathcal N_{\rm t}.
\label{eq:tx_scalar_update}
\end{align}
Both $\mathsf R_{\rm c}(x_{{\rm t},q})$ and $\Delta_{\rm B}(x_{p,q})$ in~\eqref{eq:tx_scalar_obj} admit closed-form element-wise dependencies through the SNR decomposition in~\eqref{eq:gamma_element} and the BFIM structure in~\eqref{eq:Fo_mm_elementwise}. 

After obtaining the updated $\mathbf x_{\rm t}^{(l+1)}$, ${\bf x}_{\rm r}$ is refined while keeping $\mathbf x_{\rm t}$ fixed. Similarly, the univariate objective is defined as follows:
\vspace{-5pt}\begin{align}
\phi^{(l)}_{{\rm r},q}(x_{{\rm r}, q})
&\triangleq 
- \lambda^{(l)}\,\Delta_{\rm B}(x_{{\rm r}, q})
- \frac{\rho^{(l)}}{2}\Big[\Delta_{\rm B}(x_{{\rm r}, q})\Big]_+^2,
\label{eq:rx_scalar_obj}
\end{align}
which captures the contribution of the $q$th Rx-PA to the AL function. The update of $x_{{\rm r}, q}$ is then obtained by solving the following problem:
\begin{align}
x_{{\rm r},q}^{(l+1)} = \argmax_{x_{{\rm r}, q}\in\mathcal S_{{\rm r}, q}}
\phi^{(l)}_{{\rm r},q}(x_{{\rm r}, q}),
\qquad q\in \mathcal N_{\rm r}.
\label{eq:rx_scalar_update}
\end{align}
By applying this element-wise update sequentially to all receive PAs, we obtain the updated $\mathbf x_{\rm r}^{(l+1)}$.

\subsubsection{Optimization of Multiplier and Penalty}
After completing the element-wise AO refinement of the PA layout 
${\bf X}$, the Lagrange multiplier is updated via the projected ascent rule as follows:
\begin{align}\label{eq:PS_lambda_update_only}
\lambda^{(l+1)}
= \Big[\lambda^{(l)}+\rho^{(l)}\Delta_{\rm B}\big({\bf X}^{(l+1)}\big)\Big]_+.
\end{align}
The penalty parameter is adjusted according to a monotonic schedule as follows:
\begin{align}\label{eq:penalty_update}
\rho^{(l+1)}=
\begin{cases}
\beta\rho^{(l)}, & \text{if } \big[\Delta_{\rm B}({\bf X}^{(l+1)})\big]_+ > \varepsilon_{\rm feas},\\[0.2em]
\rho^{(l)}, & \text{otherwise},
\end{cases}
\end{align}
where $\beta>1$ is a scaling factor and $\varepsilon_{\rm feas}>0$ denotes the feasibility tolerance. The outer iteration terminates once both the rate improvement $\big|\mathsf R_{\rm c}({\bf x}_{\rm t}^{(l+1)})-\mathsf R_{\rm c}({\bf x}_{\rm t}^{(l)})\big|$ and the constraint violation $\big[\Delta_{\rm B}({\bf X}^{(l+1)})\big]_+$ fall below thresholds.

Algorithm 2 summarizes the complete procedure. The computational complexity scales as $\mathcal{O}(I_{\rm out} I_{\rm in} (N_{\rm t}+N_{\rm r}) L)$, where $I_{\rm out}$ and $I_{\rm in}$ denote the iteration number. Since each element-wise AO step monotonically improves the AL function over a finite discrete domain, the inner iterations converge in finitely many steps. The adaptive penalty update further enforces feasibility, which ensures convergence to a stationary solution of the AL-reformulated problem.
\vspace{-15pt}
\subsection{Pareto-Optimal Design}
\begin{algorithm}[t]
\caption{Element-wise AO Algorithm for the Rate-Profile Based Pareto-Optimal Design}
\label{alg:PS_pareto_scan_AO}
\begin{algorithmic}[1]
\REQUIRE Profile set $\{\alpha_1,\!\ldots,\!\alpha_{N_\alpha}\!\}\!\!\subset\!(0,1)$, initial layout $\mathbf X^{(0)}$.
\FOR{$n=1,2,\ldots,N_\alpha$}
  \STATE Set $\alpha\gets \alpha_n$, $l\gets 0$, and initialize $\mathbf X^{(0)}(\alpha)\gets \mathbf X^{(0)}$.
  \REPEAT
    \FOR{$q\in\mathcal N_{\rm t}$}
      \STATE Construct $\mathcal S_{{\rm t},q}$ and update $x_{{\rm t},q}$ by solving~\eqref{eq:Pareto_Tx_Update_refined}.
    \ENDFOR
    \FOR{$q\in\mathcal N_{\rm r}$}
      \STATE Construct $\mathcal S_{{\rm r},q}$ and update $x_{{\rm r},q}$ by solving~\eqref{eq:Pareto_Rx_Update_refined}.
    \ENDFOR
    \STATE $l\gets l+1$.
  \UNTIL{convergence \textbf{or} $l=I_{\rm iter}$}
  \STATE Obtain $\mathbf X^\star(\alpha)\gets \mathbf X^{(l)}$ and compute $r^\star(\alpha)$ in~\eqref{eq:slot_value}.
  \STATE Record the Pareto point $\big(\mathsf{BCRB}(\mathbf X^\star(\alpha)),\ \mathsf R_{\rm c}(\mathbf x_{\rm t}^\star(\alpha))\big)$.
\ENDFOR
\STATE \textbf{Output:}
$\{(\mathsf{BCRB}(\mathbf X^\star(\alpha_n)),\,\mathsf R_{\rm c}(\mathbf x_{\rm t}^\star(\alpha_n)))\}_{n=1}^{N_\alpha}$.
\end{algorithmic}
\end{algorithm}

Following the rate-profile formulation in Section~\ref{sec:Pareto_sPA}, we characterize the sensing-communications tradeoff in the multi-PA setting by optimizing the PA locations for a given profile parameter $\alpha\in[0,1]$.
Define the sensing rate as $\mathsf{R}_{\rm s}(\mathbf X)\triangleq[\mathsf{BCRB}(\mathbf X)]^{-1}$.
The corresponding scalarized problem is
\begin{align}\label{eq:PS_rate_profile_master}
\max_{\mathbf X\in\mathcal X,\ \mathsf R}\ \ \mathsf R, \quad  \text{s.t.} \ \mathsf R_{\rm c}(\mathbf x_{\rm t}) \ge \alpha \mathsf R, \ \mathsf R_{\rm s}(\mathbf X) \ge (1-\alpha)\mathsf R.
\end{align}
Equivalently, we maximize the rate-profile utility
\begin{align}\label{eq:slot_value}
r^\star(\alpha)
\triangleq 
\max_{\mathbf X\in\mathcal X}
\min\!\left\{
\frac{\mathsf R_{\rm c}(\mathbf x_{\rm t})}{\alpha+\delta(\alpha,0)},
\frac{\mathsf R_{\rm s}(\mathbf X)}{(1-\alpha)+\delta(\alpha,1)}
\right\}.
\end{align}

The above problem is solved by a sequential element-wise AO procedure, where each PA is updated via a one-dimensional search over its local feasible set while keeping all other PAs fixed.
In particular, the transmit-side update accounts for both communications and sensing through the rate-profile objective, whereas the receive-side update only affects sensing.

\subsubsection{Tx-PA Update}
For each $q\in\mathcal N_{\rm t}$, the $q$th Tx-PA is updated by solving
\begin{align}\label{eq:Pareto_Tx_Update_refined}
x_{{\rm t},q}^{\star} \!\!= \!
\argmax_{x\in\mathcal S_{{\rm t},q}}
\min\!\left\{
\frac{\mathsf R_{\rm c}(x)}{\alpha+\delta(\alpha,0)},
\frac{\mathsf R_{\rm s}(x)}{(1-\alpha)+\delta(\alpha,1)}
\right\}.
\end{align}

\subsubsection{Rx-PA Update}
For each $q\in\mathcal N_{\rm r}$, the $q$th Rx-PA is updated by maximizing the sensing rate, i.e.,
\begin{align}\label{eq:Pareto_Rx_Update_refined}
x_{{\rm r},q}^{\star} =
\argmax_{x\in\mathcal S_{{\rm r},q}}
\mathsf R_{\rm s}(x) =
\argmin_{x\in\mathcal S_{{\rm r},q}}
\mathsf{BCRB}(x).
\end{align}

By cyclically sweeping over all Tx- and Rx-PAs and repeating the above updates until convergence, we obtain a locally optimal solution for the chosen $\alpha$.
Since each update is selected from a finite discrete set and does not decrease the rate-profile utility in~\eqref{eq:slot_value}, the algorithm converges in a finite number of iterations.
The computational complexity is $\mathcal{O}\big(I_{\rm iter}(N_{\rm t}+N_{\rm r})\,L\,(K+T^{2})\big)$, where $I_{\rm iter}$ denotes the number of AO iterations.

\vspace{-10pt}
\section{Numerical Results}
This section presents numerical results to demonstrate the advantages of PASS and to validates the effectiveness of our proposed optimization algorithms. Unless stated otherwise, the following setup is adopted.

We consider two dielectric waveguides that are parallel to the $x$-axis and located at $y_{\rm t}=3$~m and $y_{\rm r}=-3$~m, respectively. Following the discretized waveguide model in the previous sections, each waveguide is uniformly discretized into $L=1000$ candidate PA locations.
The service region is specified by $h=5$~m, $D_x=10$~m, and $D_y=6$~m. In each realization, the $K=4$ multicast users are independently and uniformly distributed over $x\in[0,D_x]$ and $y\in[-D_y/2,\,D_y/2]$. A single sensing target is considered, 
the mean $\boldsymbol\mu$ is independently and uniformly drawn from the same rectangular region, while the prior variances are independently generated as $\sigma_x^2\sim\mathcal U(0,1)$ and $\sigma_y^2\sim\mathcal U(0,2)$.

The carrier frequency is set to $f_{\rm c} = 28$~GHz,
 the noise power at both the users and the target is $\sigma_k^2=\sigma_{\rm s}^2=-90$~dBm. For the waveguide propagation, the guided wavelength is set to $\lambda_{\rm g}=\lambda/1.44$~\cite{10945421}. The minimum SNR requirement is $\gamma_{\rm c}=12$~dB, and the maximum BCRB requirement is $\Gamma_{\rm s}=0.1$.
For the proposed AL-based AO algorithm, the penalty factor is initialized as $\rho=10^{-4}$ and updated with scaling parameter $\beta=2$, and the convergence tolerance is set to $10^{-3}$. All statistical results are averaged over $1000$ random realizations.

For performance comparison, we evaluate the proposed design against the following benchmark schemes: 1) {Fixed-array with analog beamforming (BF):} Fixed-location uniform linear arrays (ULAs) with $N=6$ elements are deployed at $(0,y_{\rm t},h)$ and $(0,y_{\rm r},h)$. A single RF chain is employed, and the analog BF satisfy the constant-modulus constraint. 2) {Fixed-array with digital BF:} The ULAs are deployed at the same locations, but fully digital precoding is assumed, which is based on maximum-ratio transmission/combining (MRT/MRC) towards the multicast users and the sensing target. This architecture serves as an upper bound for fixed-array baselines. 3) {Random PA:} The $N$ PAs are randomly placed over the waveguide, subject to the minimum inter-element spacing constraint. 4) {Centered PA:} The $N$ PAs are placed contiguously at the waveguide center with the minimum inter-element spacing $\Delta_{\min}$.

\vspace{-5pt}
\subsection{Single-PA Case}
\subsubsection{Multicast Communication Performance}
Fig. \ref{fig:Rate_singlePA} illustrates the achievable multicast rate versus the side length $D_{\rm x}$. As the service region expands, the multicast rate exhibits a monotonic decline across all evaluated schemes. This trend is physically attributed to the severe path-loss attenuation in larger areas, which diminishes the ability of a single Tx-PA to simultaneously maintain robust LoS links to all distributed users. Nevertheless, the proposed C-C design consistently yields the highest multicast rate, which significantly outperforms both the random PA and centered PA layouts. This superiority validates the effectiveness of our closed form solution for the C-C design.
Furthermore, it is worth noting that the performance degradation of the C-C design is less pronounced compared to other schemes. This also demonstrates the robustness of the proposed C-C design in spatially extended scenarios.
\begin{figure}[!t]
\centering
\includegraphics[height=0.26\textwidth]{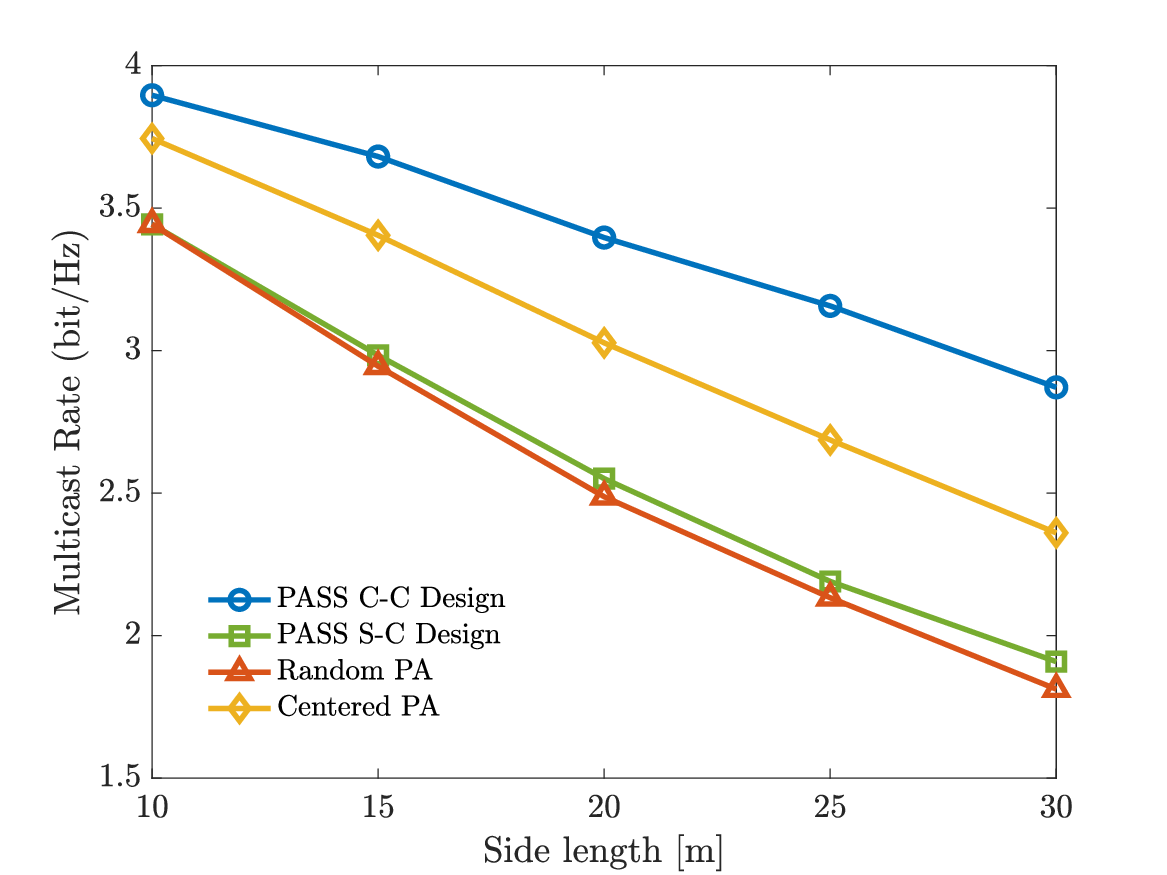}
\caption{Multicast rate versus the side length $D_{\rm x}$ for single-PA case.}\label{fig:Rate_singlePA}
\vspace{-5pt}
\end{figure}
\subsubsection{Sensing Performance}
\begin{figure}[!t]
\centering
\includegraphics[height=0.26\textwidth]{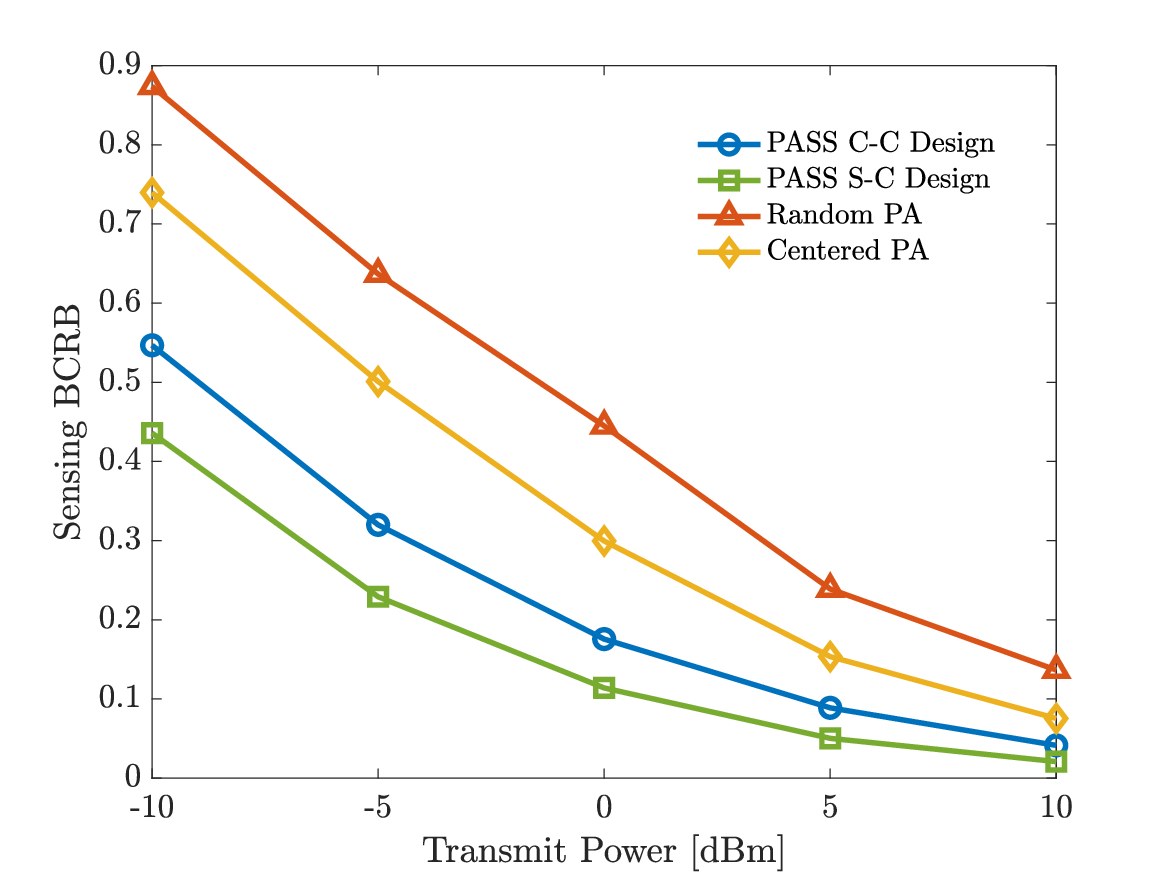}
\caption{BCRB versus the transmit power for single-PA case.}\label{fig:Pt_BCRB}
\vspace{-5pt}
\end{figure}

\begin{figure}[!t]
\centering
\includegraphics[height=0.26\textwidth]{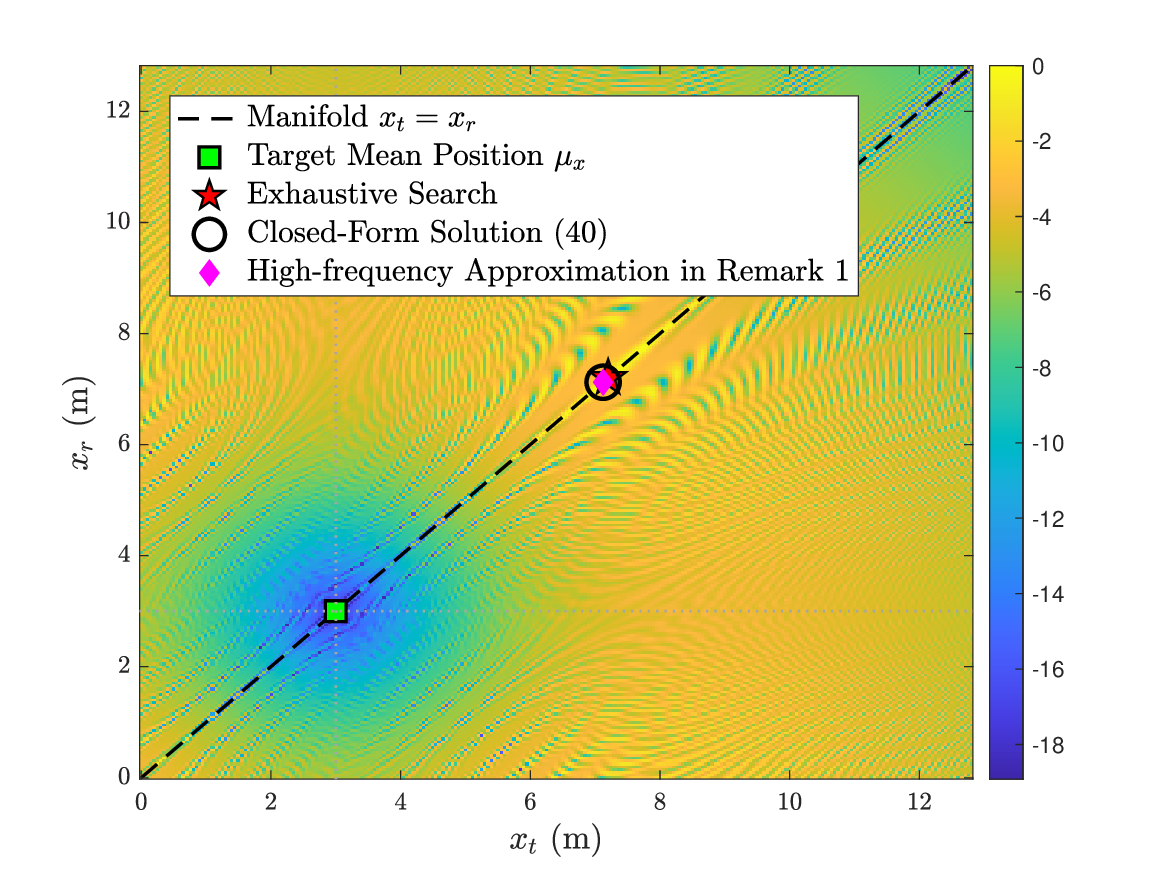}
\caption{The optimal Tx/Rx-PA placement for maximizing the OFIM (dB) under the proposed closed-form solution and the exhaustive search benchmark.}
\label{fig:optimal_placement}
\vspace{-5pt}
\end{figure}
Fig. \ref{fig:Pt_BCRB} depicts the sensing BCRB performance as a function of the transmit power. The results reveals several insightful observations. First, a monotonic decrease in BCRB is observed for all schemes, which is consistent with the theoretical improvement in estimation accuracy at higher SNRs. Second, the proposed S-C design consistently yields the lowest BCRB across the entire power regime. This performance advantage serves as a benchmark and validates the optimality of the derived closed-form solution for S-C design.
Third, it is noteworthy that the C-C design yields superior sensing performance compared to the random PA benchmark, even though its PA placement is optimized solely for the multicast rate without explicitly accounting for sensing metrics. This performance gain is attributed to the specific design strategy employed: since the Rx-PA is independent of the multicast rate, we deliberately adopted the transceiver symmetric layout for the Rx-PA deployment in the C-C design. Consequently, compared to the random PA scheme where both the transceiver are placed stochastically, the C-C design benefits significantly from this structural symmetry. This observation not only demonstrates the robustness of the C-C design but also explicitly validates the effectiveness of the proposed symmetric layout strategy in Proposition 1.

To rigorously validate the global optimality of the derived closed-form solution for the C-C design, Fig. \ref{fig:optimal_placement} plots the spatial distribution of the OFIM w.r.t the Tx- and Rx-PA placement. The visualization reveals the highly non-convex and oscillating nature of the sensing objective function, characterized by numerous local optima arising from the rapid signal phase variations. Despite this complex optimization landscape, the proposed closed-form solution in \eqref{eq:closed_form_d} aligns perfectly with the global optimal position obtained via exhaustive search. Furthermore, the high-frequency approximation solution in Remark~1 also lies in close proximity to the global optimum, which demonstrates the tightness of the derived approximation expression. The results confirm that the derived close-form expressions can efficiently pinpoint the optimal PA locations for maximizing sensing performance without suffering from the prohibitive computational complexity of grid search or the risk of trapping in local optima.

\subsubsection{Pareto-Optimal Performance}
\begin{figure}[!t]
\centering
\includegraphics[height=0.26\textwidth]{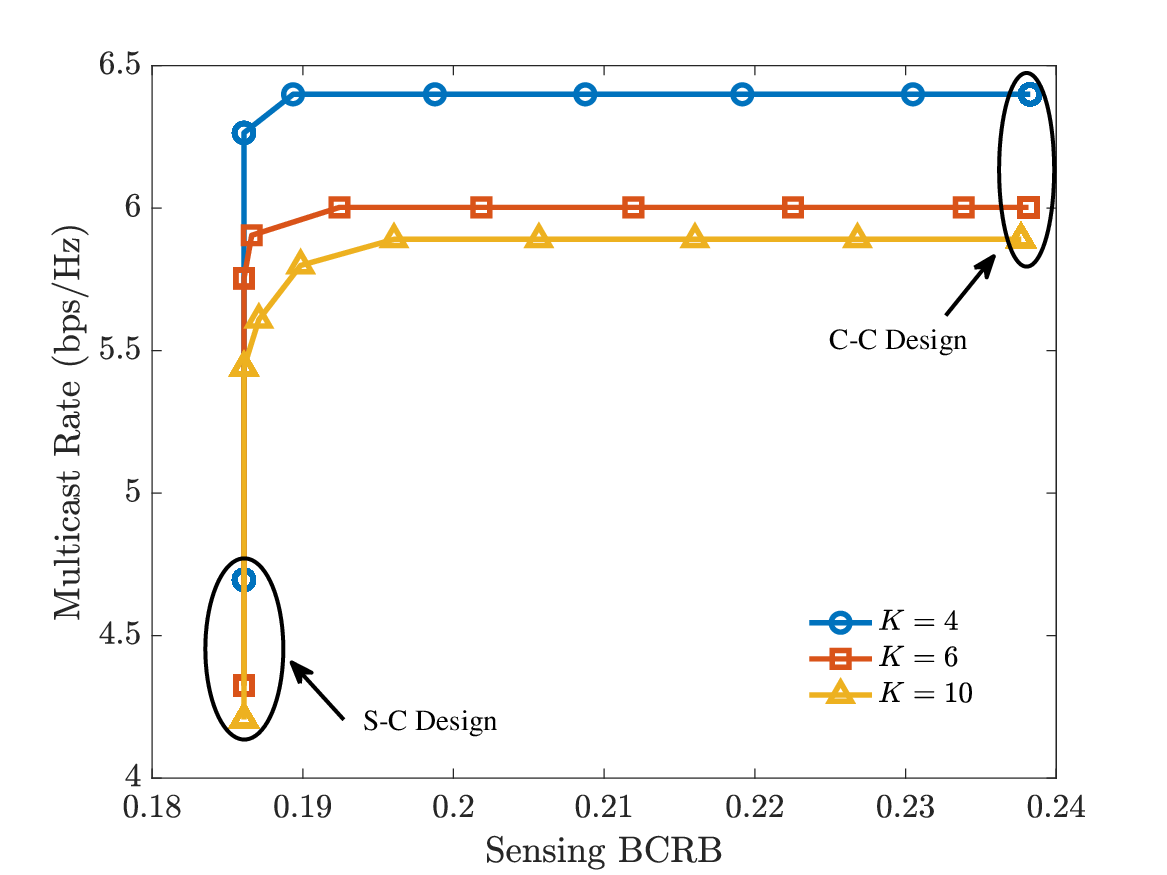}
\caption{Patero-optimal design for single-PA case.}
\label{fig:pareto_singlePA}
\vspace{-5pt}
\end{figure}
Fig. \ref{fig:pareto_singlePA} illustrates the Pareto frontier of the achievable multicast rate versus the sensing BCRB. The results reveal the fundamental tradeoff limits and robust performance characteristics of the proposed architecture.
First, the extended saturation region of the multicast rate indicates the spatial stationarity of the communication channel. In the single-PA case, the multicast performance is predominantly limited by the large-scale path loss of the worst-case user. Therefore, the objective function remains quasi-constant over local spatial displacements. This implies that the communication rate is robust to minor variations in the PA position. 
Second, the vertical region demonstrates the high robustness of the sensing objective near its global optimum. Although the S-C design yields the optimal PA placement for sensing, the BCRB function is locally flat around this optimal point. Consequently, the system can deviate slightly from the strict S-C solution to move closer to the user cluster. This slight deviation incurs a negligible degradation in BCRB but translates into a significant improvement in the multicast rate.

\subsection{Multi-PA Case}
\subsubsection{Convergence Analysis of the Proposed Algorithms}
\begin{figure}[!t]
    \centering
    \includegraphics[height=0.5\linewidth]{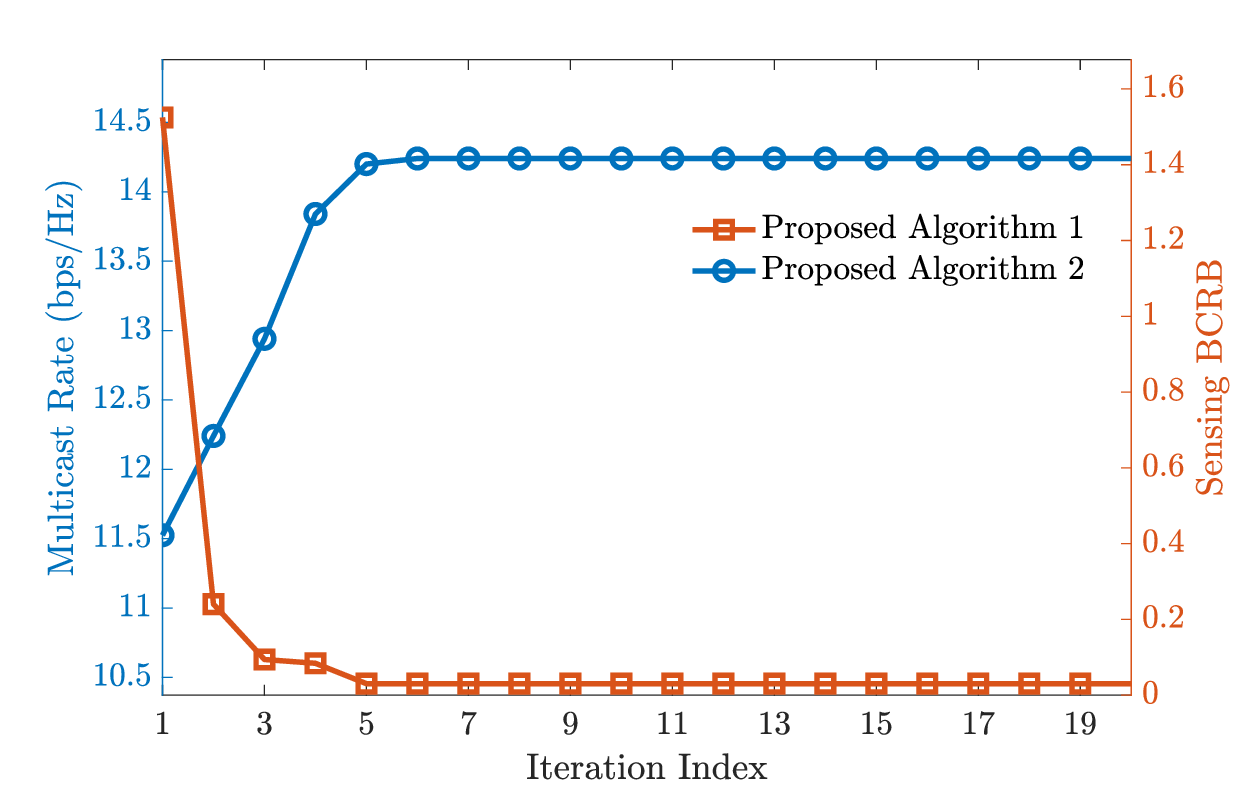}
    \caption{Convergence behavior of the proposed Algorithms. The left and right vertical axes represent the multicast rate (blue curve) obtained by Algorithm 2 and the sensing BCRB (orange curve) optimized by Algorithm 1, respectively.}
    \label{fig:Convergence} 
\vspace{-5pt}
\end{figure}
Fig. \ref{fig:Convergence} illustrates the convergence behavior of the proposed algorithms. First, regarding Algorithm 1, the objective function exhibits a rapid and monotonic decline, which stabilizes within a few iterations. This demonstrates the robustness of the proposed AO strategy. Despite the highly non-convex and oscillating sensing landscape as revealed in Fig. \ref{fig:optimal_placement}, the algorithm effectively navigates the local fluctuations to identify a high-quality solution satisfying the communication constraints. Second, regarding Algorithm 2, the multicast rate also shows a fast convergence trajectory. This efficiency is attributed to the adopted AL framework. By incorporating the distinct constraints into the objective function via penalty terms, the AL method transforms the original constrained problem into a sequence of unconstrained sub-problems. This effectively smooths the optimization landscape and avoids the computational oscillation often encountered near the boundaries of the feasible set, thereby accelerating the convergence speed.

\begin{figure}[!t]
\centering
\includegraphics[height=0.26\textwidth]{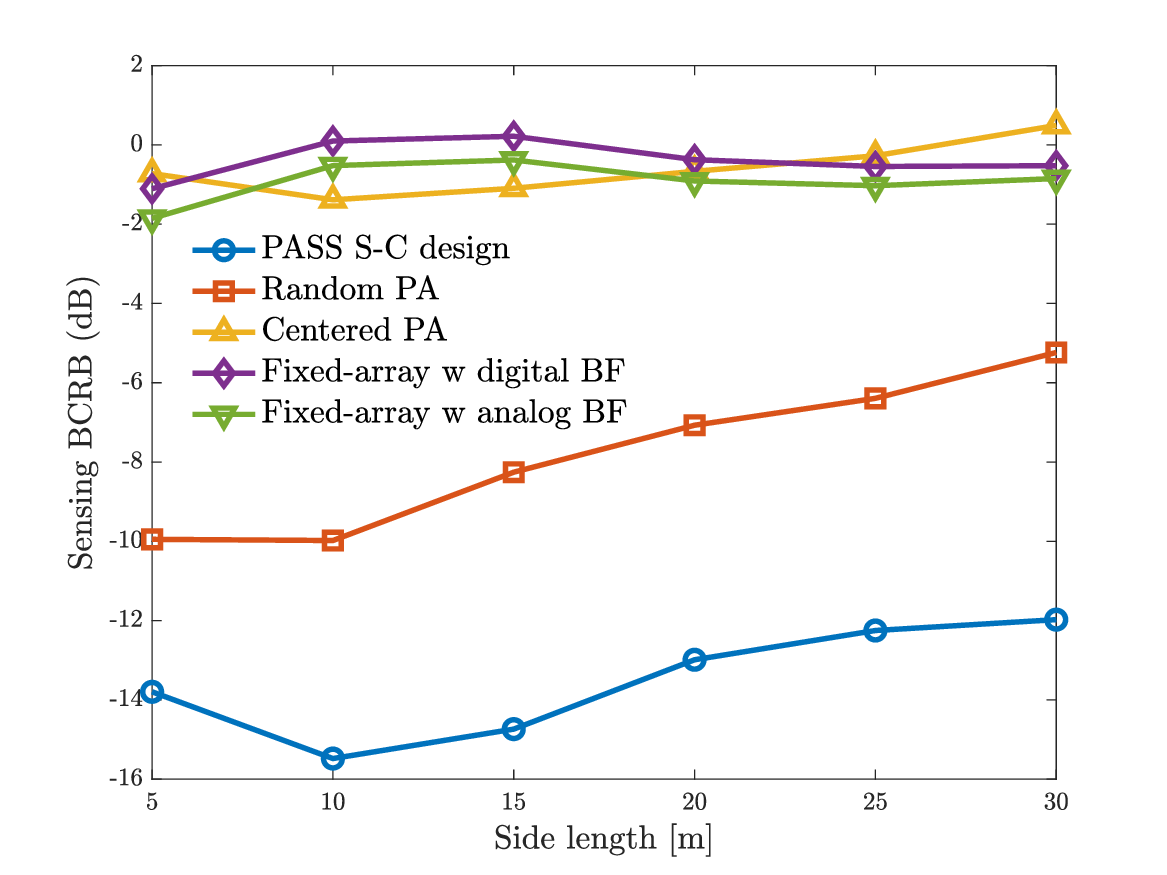}
\caption{BCRB (dB) versus the side length $D_{\rm x}$ for multi-PA case.}\label{fig:Dx_BCRB}
\vspace{-5pt}
\end{figure}
\subsubsection{Sensing Performance}
Fig.~\ref{fig:Dx_BCRB} illustrates the sensing performance in terms of the BCRB versus the service region. 
The proposed S-C design consistently achieves the superior performance across all region sizes, which demonstrates the effectiveness of the proposed element-wise AO framework. Notably, the BCRB exhibits a non-monotonic trend that initially decreases before increasing as the service region expands. This behavior can be rationalized by examining the interplay between geometric freedom and communication constraints.
In the regime of small service regions, the finite length of the waveguide restricts the lateral offset of PAs, which prevents them from attaining the optimal geometric placement for sensing. This spatial restriction leads to a suboptimal observation geometry and hence degraded sensing accuracy. 
As the deployment region expands, the PAs' layout becomes more spatially flexible, which allows more effective sensing design.
However, when the service region becomes excessively large, the multicast rate constraint becomes increasingly difficult to satisfy, as the users are more widely distributed in the development region. 
Consequently, the stricter multicast constraint necessitates a wider spatial dispersion of the PAs to ensure adequate coverage for all multicast users, which restricts the optimization of the BCRB and leads to a gradual degradation in sensing performance.
In comparison, the random PA layout performs about 4~dB worse than the proposed S-C design, as it fails to exploit the array gain. 
Nonetheless, its distributed PA positions offer a larger effective aperture compared to the fixed-array and centered PA layout benchmarks, whose compact geometries severely limit the aperture gain. 
Overall, these results highlight that the sensing advantage of PASS arises not only from its inherently large aperture but also from the array gain achieved via the dynamic reconfiguration of PA placement.

\subsubsection{Multicast Performance}
\begin{figure}[!t]
\centering
\includegraphics[height=0.26\textwidth]{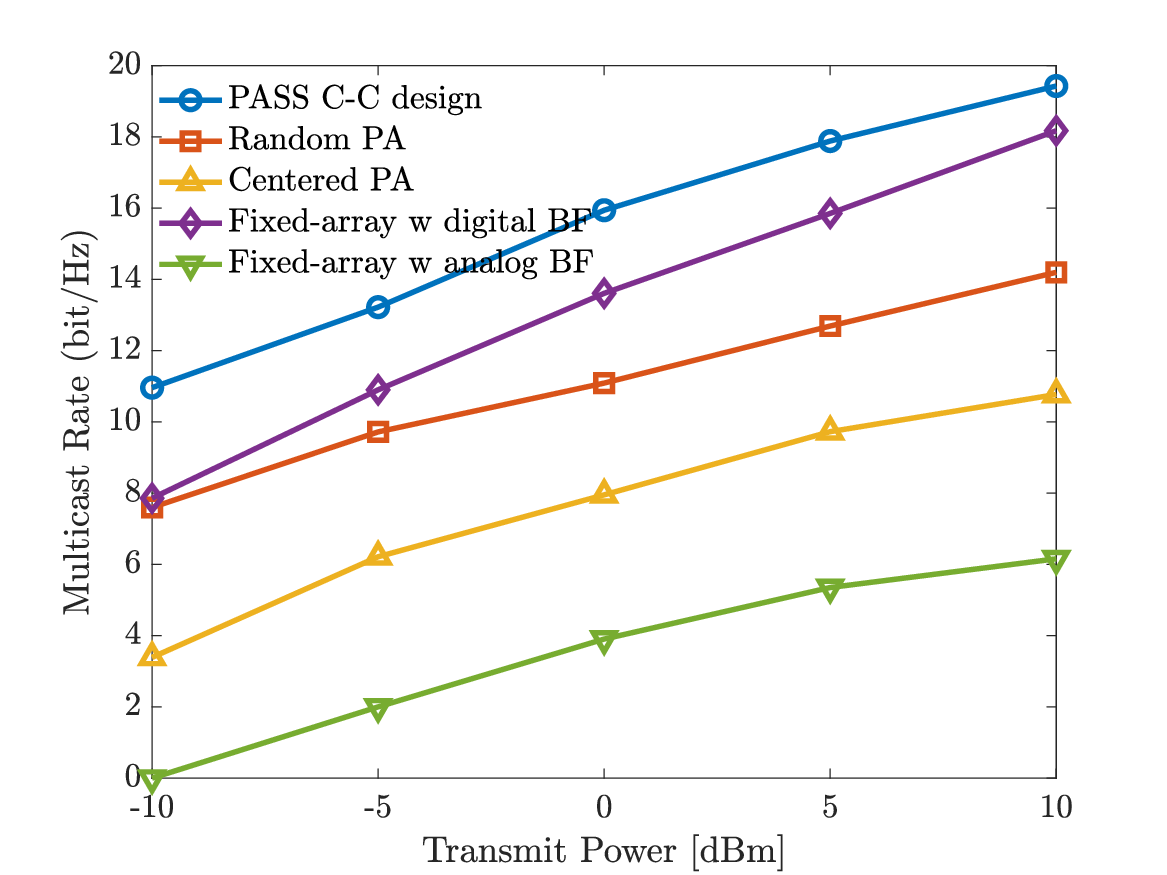}
\caption{Multicast rate versus the transmit power for multi-PA case.}\label{fig:Rate_Power}
\end{figure}
Fig.~\ref{fig:Rate_Power} shows the multicast rate versus the transmit power for multi-PA case.
The proposed C-C design achieves the highest multicast rate across all power levels, which demonstrates the effectiveness of the AL-based element-wise AO algorithm in concentrating energy toward the users while satisfying the sensing constraint. The fixed-array scheme with digital BF achieves the second-best performance, which is attributed to its high precoding flexibility. However, this performance comes at the cost of high hardware complexity and power consumption.
Notably, the random PA layout outperforms the centered PA benchmark. This advantage arises because the spatially distributed PAs are more likely to be located in the vicinity of dispersed users, thereby establishing strong LoS links and shortening communication distances. In contrast, the compact PA placement of the centered layout and fixed-array schemes restrict such proximity gains.
Consequently, the proposed PASS architecture serves as a superior solution that simultaneously achieves high spectral efficiency and maintains low hardware complexity.

\subsubsection{Pareto-Optimal Performance}
\begin{figure}[!t]
\centering
\includegraphics[height=0.26\textwidth]{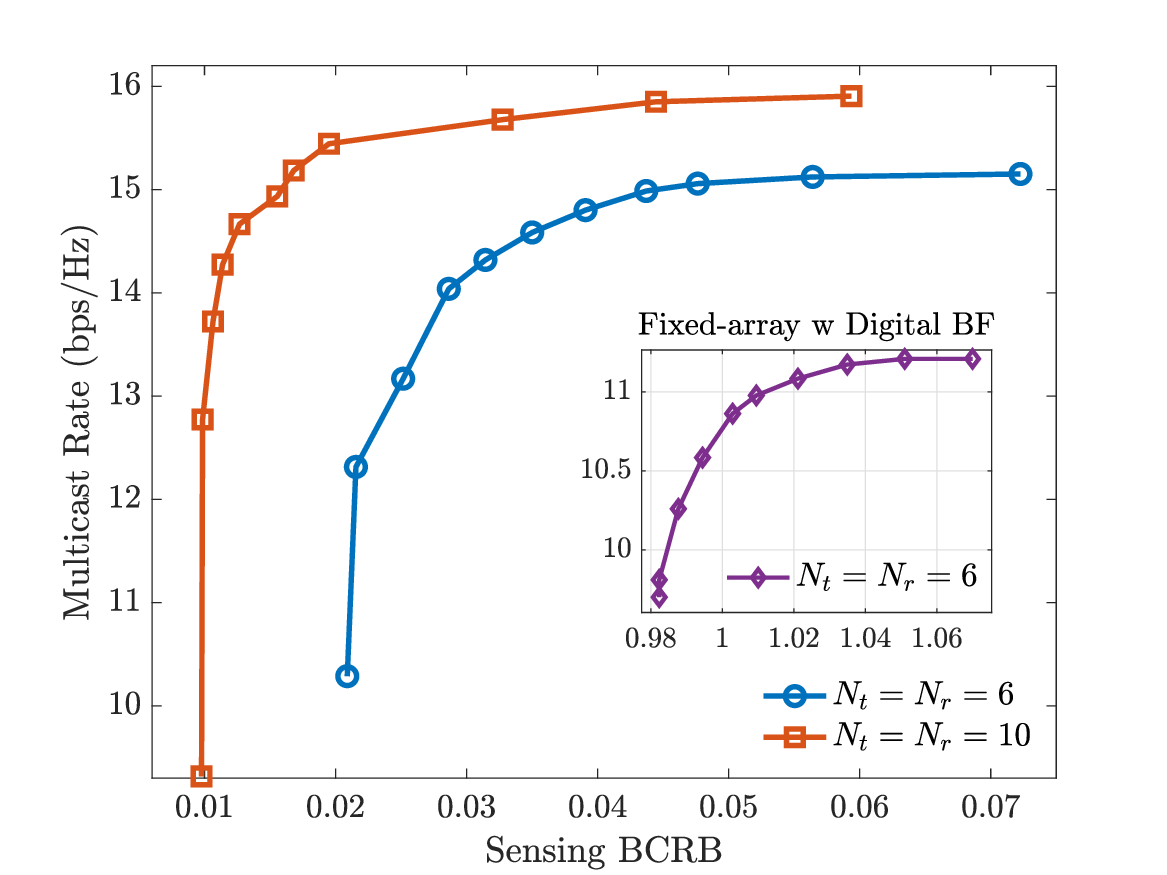}
\caption{Pareto-optimal design for multi-PA case.}
\label{fig:pareto_multiPA}
\end{figure}
Fig. \ref{fig:pareto_multiPA} illustrates the Pareto frontier of the proposed PASS architecture compared with the fixed-array digital BF benchmark. Notably, increasing the number of PAs significantly expands the achievable Pareto region. This expansion indicates that the conflict between multicast rate and sensing accuracy is progressively mitigated as the system scale grows. This gain arises from the higher spatial flexibility provided by the larger array. With more PAs available, the PASS architecture can dedicate specific subsets of PAs to satisfy the sensing resolution requirements, while simultaneously maintaining sufficient multicast rate for the users. In contrast, the digital BF benchmark exhibits a constrained tradeoff profile. Despite its high precoding flexibility, its fixed topology imposes a rigid geometric bottleneck, which limits its ability to simultaneously minimize BCRB and maximize rate compared to the geometrically adaptive PASS.

\vspace{-5pt}
\section{Conclusion}

This paper proposed a PASS-aided ISAC framework that exploits spatial flexibility to jointly support information multicasting and target localization. Within this framework, we derived the BCRB to quantify sensing accuracy using GHQ rules, and we optimized the PA placement under C-C, S-C, and Pareto-optimal criteria.
For the single-PA case, we obtained a closed-form solution for the C-C transmit PA, while we established a closed-form symmetric transceiver layout for the S-C design.
For the multi-PA scenario, we developed a sequential element-wise AO method to tackle the non-convex multivariate PA placement problems. We applied this strategy to directly solve the S-C design and further integrated it into an AL framework and a rate-profile formulation to address the C-C and Pareto-optimal designs, respectively.
Numerical results demonstrated that the proposed PASS architecture substantially outperforms conventional fixed-array in both multicasting and sensing. 
Most insightfully, our analysis showed that improving sensing accuracy relies not on power scaling, but on optimizing the geometric offset induced by the PA placement. This spatial reconfiguration capability is unique to PASS and exploits the offset to achieve superior performance without compromising the communication rate.

\appendices
\vspace{-5pt}
\section*{Appendix A: Proof of Proposition 1}
Starting from the expression of the conditional Fisher information in~\eqref{eq:Fxx_exact_compact}, we rewrite it as a product of (i) a \emph{distance-dependent} term and (ii) a bounded \emph{geometry-dependent} term:
\begin{align}\label{eq:obj_fxx}
	F_{xx}  = {2|C|^2}{\sigma_{\rm s}^{-2}} {\hat F}_{xx} {\check F}_{xx},
\end{align}
where ${\hat F}_{xx}\triangleq({R_{\rm t}^2 R_{\rm r}^2})^{-1}$ captures the path-loss decay proportional to $|\mu|^2$, while ${\check F}_{xx}\triangleq \mathcal{T}_{\rm self} + \mathcal{T}_{\rm cross}\Phi$ encapsulates the bracketed angular terms in \eqref{eq:Fxx_exact_compact}, with $\mathcal{T}_{\rm self} \triangleq \sum_{p\in\{{\rm t,r}\}}(k_0^2 + R_p^{-2})\cos^2\theta_p$ and $\mathcal{T}_{\rm cross} \triangleq 2(k_0^2 + (R_{\rm t}R_{\rm r})^{-1})\cos\theta_{\rm t}\cos\theta_{\rm r}$.

\subsubsection{Maximization of the Distance-Dependent Factor}
We first analyze the dominant term ${\hat F}_{xx}$. Let $m$ denote the fixed average lateral offset of the PAs from the target mean $\mu_x$. We parameterize the asymmetry by $\varepsilon \ge 0$, such that the individual lateral offsets are $\delta_{\rm t} = |\mu_x - x_{\rm t}| = m + \varepsilon$ and $\delta_{\rm r} = |\mu_x - x_{\rm r}| = m - \varepsilon$. The squared slant ranges are $R_{\rm t}^2 = (m+\varepsilon)^2 + \Delta_{\rm s}^2$ and $R_{\rm r}^2 = (m-\varepsilon)^2 + \Delta_{\rm s}^2$.
Expanding ${\hat F}_{xx}^{-1}$ w.r.t. $\varepsilon$ yields
\begin{align}
	{\hat F}_{xx}^{-1}(\varepsilon) &\!=\! \left[ (m^2+\Delta_{\rm s}^2+\varepsilon^2) + 2m\varepsilon \right] \left[ (m^2+\Delta_{\rm s}^2+\varepsilon^2)\! -\! 2m\varepsilon \right] \notag \\
	&\!=\! {c}_0 + 2(\Delta_{\rm s}^2 - m^2)\varepsilon^2 + \varepsilon^4,
\end{align}
where ${c}_0 = (m^2+\Delta_{\rm s}^2)^2$ is a constant. In the typical sensing region where the PA-target distance $\Delta$ exceeds the lateral PA offset $m$ (i.e., $\Delta_{\rm s} > m > 0$), the coefficient of the quadratic term is positive. Therefore, ${\hat F}_{xx}^{-1}(\varepsilon)$ is a strictly monotonically increasing function of $\varepsilon$. Consequently, the intensity factor ${\hat F}_{xx}(\varepsilon)$ achieves its unique global maximum at $\varepsilon=0$, corresponding to the symmetric configuration. 

\subsubsection{Boundedness of the Geometry-Dependent Factor}
Next, we examine the geometric factor ${\check F}_{xx}$. The Bayesian objective involves the prior expectation $\widetilde F_{xx}(x_{\rm t},x_{\rm r})=\mathbb E_{u^x}[F_{xx}(u^x,x_{\rm t},x_{\rm r})]$. In \eqref{eq:Fxx_exact_compact}, the only sign-changing term is the oscillatory factor $\Phi=\cos\!\big(k_0(R_{\rm t}-R_{\rm r})\big)$ in the cross component. When the prior of $u^x$ is non-degenerate (i.e., $\sigma_x^2>0$), the range difference $R_{\rm t}-R_{\rm r}$ varies with $u^x$, so the phase $k_0(R_{\rm t}-R_{\rm r})$ sweeps over multiple cycles in typical high-frequency regimes. As a result, the cross contribution tends to be attenuated after taking expectation due to phase averaging, i.e., $\mathbb E_{u^x}[\mathcal T_{\rm cross}\Phi]$ is much less sensitive than its pointwise value. Therefore, increasing $\widetilde F_{xx}$ is equal to maximize the path-loss factor $|\mu|^2$. Since $(R_{\rm t}^2R_{\rm r}^2)^{-1}$ is uniquely maximized at $\varepsilon=0$ for $\Delta_{\rm s}>m$, the Fisher information favors the equal-offset geometry $\varepsilon=0$, i.e., $|u^x-x_{\rm t}|=|u^x-x_{\rm r}|$. This concludes the proof.

\vspace{-15pt}
\section*{Appendix B: Proof of Proposition~2}
Let $u^x=\mu_x+\delta$, where $\delta$ follows a symmetric prior with $\mathbb E[\delta]=0$ and $\mathbb E[\delta^2]=\sigma_x^2$ i.e., $p(\delta)=p(-\delta)$. Consider the same-side equal-offset parameterization $x_{\rm t}=x_{\rm r}=c-d$ with $d\ge 0$, and define $\widetilde F_{xx}(c,d)\triangleq \mathbb E_{\delta}\!\left[F_{xx}(\mu_x+\delta,c,d)\right]$.

For a fixed $d$, shifting $c$ only changes the relative offset between the target and the PAs. Hence, there exists a smooth function $f(\cdot)$ such that
$F_{xx}(\mu_x+\delta,c,d)=f\!\left(\delta-(c-\mu_x)\right)$, which yields
\begin{align}
\partial_c \widetilde F_{xx}(c,d)
= -\,\mathbb E_{\delta}\!\left[f'\!\left(\delta-(c-\mu_x)\right)\right].
\end{align}
Evaluating at $c=\mu_x$ gives $\partial_c \widetilde F_{xx}(c,d)\big|_{c=\mu_x}=-\mathbb E_{\delta}[f'(\delta)]$.
Moreover, at $c=\mu_x$ the geometry is symmetric w.r.t. $\delta\mapsto-\delta$, implying $f(\delta)=f(-\delta)$ and thus $f'(\delta)=-f'(-\delta)$.
Using $p(\delta)=p(-\delta)$, we obtain $\mathbb E_{\delta}[f'(\delta)]=0$, and therefore
$\partial_c \widetilde F_{xx}(c,d)\big|_{c=\mu_x}=0$.
This proves that $c=\mu_x$ is a stationary point under a symmetric prior.

After fixing $c=\mu_x$, define $G(u^x,d)\triangleq \partial_d F_{xx}(u^x,d)$.
Interchanging differentiation and expectation gives
\begin{align}
\frac{{\rm d}}{{\rm d}d}\,\mathbb E_{u^x}\!\big[F_{xx}(u^x,d)\big]
= \mathbb E_{\delta}\!\big[G(\mu_x+\delta,d)\big].
\end{align}
Applying a second-order Taylor expansion of $G(\mu_x+\delta,d)$ around $\delta=0$ yields
\begin{align}
G(\mu_x+\delta,d)
& =G(\mu_x,d)+\partial_{u^x}G(\mu_x,d)\delta\notag \\
& +\frac{1}{2}\partial_{u^x u^x}G(\mu_x,d)\delta^2 +\mathcal O(|\delta|^3).
\end{align}
Taking expectation and using $\mathbb E[\delta]=0$ and $\mathbb E[\delta^2]=\sigma_x^2$ gives
\begin{align}
&\mathbb E_{\delta}\!\big[G(\mu_x+\delta,d)\big]\notag \\
& 
=G(\mu_x,d)+\frac{\sigma_x^2}{2}\partial_{u^x u^x}G(\mu_x,d)+\mathcal O(\mathbb E[|\delta|^3]).
\end{align}
For a Gaussian prior, $\mathbb E[|\delta|^3]=\Theta(\sigma_x^3)$, hence the remainder term is $\mathcal O(\sigma_x^3)$.
Substituting back $G(u^x,d)=\partial_d F_{xx}(u^x,d)$ yields
\begin{align}
\frac{{\rm d}}{{\rm d}d}\,\mathbb E_{u^x}\!\big[F_{xx}(u^x,d)\big]
= \partial_d F_{xx}(\mu_x,d)+\mathcal O(\sigma_x^2),
\end{align}
which completes the proof.

\section*{Appendix C: Proof of Remark~1}
For brevity, let $s\triangleq R^2=d^2+\Delta^2$ and $G(s)\triangleq{(1+k_0^2 s)(s-\Delta^2)}s^{-4}$.
Since ${\rm d}s/{\rm d}d=2d>0$ for $d>0$, the stationarity conditions are equivalent, i.e., $\tfrac{{\rm d}F}{{\rm d}d}=0$ or $\tfrac{{\rm d}G}{{\rm d}s}=0$. 
Moreover, the differentiate of $G$ can be given by
\begin{align}
\frac{{\rm d}G}{{\rm d}s}
&=\Big[(1+k_0^2 s)+k_0^2(s-\Delta^2)\Big]s^{-4}\notag \\
&-4(1+k_0^2 s)(s-\Delta^2)s^{-5}=0.
\end{align}
It follows that
\begin{equation}\label{eq:s-poly}
-2k_0^2 s^2-3s+3k_0^2\Delta^2 s+4\Delta^2=0.
\end{equation}

Substituting $s=d^2+\Delta^2$ into \eqref{eq:s-poly} and defining $x\triangleq d^2$ gives the univariate quadratic as follows:
\begin{equation}\label{eq:quad-x}
-2k_0^2 x^2-\big(k_0^2\Delta^2+3\big)x+\big(k_0^2\Delta^4+\Delta^2\big)=0.
\end{equation}
Solving \eqref{eq:quad-x} and retaining the nonnegative root yields
\begin{equation}\label{eq:d2-closed}
d^{2} =\frac{-\,k_0^{2}\Delta^{2}-3+\sqrt{\,9k_0^{4}\Delta^{4}+14k_0^{2}\Delta^{2}+9\,}}{4k_0^{2}}.
\end{equation}

For $k_0\Delta\to\infty$ , a binomial expansion of the square root in \eqref{eq:d2-closed} gives
\vspace{-10pt}
\begin{equation}\label{eq:asymp-d2}
d =\sqrt{\frac{\Delta^{2}}{2}-\frac{3}{4k_0^2}+{\mathcal O}(k_0^{-4})}\approx\frac{\Delta}{\sqrt 2}.
\end{equation}
Here, this proof is completed.

\bibliographystyle{IEEEtran} 
\bibliography{reference}   

\end{document}